\documentclass[letterpaper,11pt]{article}

\usepackage{cmap}
\usepackage[utf8x]{inputenc}
\usepackage[T1]{fontenc}
\usepackage{amsthm}
\usepackage{listings}
\usepackage{multirow}
\usepackage{amsfonts}
\usepackage{parskip}
\usepackage{enumitem}
\usepackage{ytableau}
\usepackage{subcaption}
\usepackage{tabularx}
\usepackage{booktabs}

\ytableausetup{
    boxsize=1.8em,
    centertableaux
}

\usepackage{caption}
\usepackage{subcaption}

\usepackage[pdftex]{graphicx} %
\usepackage[english]{babel} %
\usepackage[colorlinks,citecolor=blue,linkcolor=red]{hyperref}
\usepackage{calc} %
\usepackage{enumitem} %
\setlist{nosep}

\usepackage[margin=0.94in]{geometry}

\usepackage[all]{nowidow} %
\usepackage[protrusion=true,expansion=true]{microtype} %

\usepackage{svg}
\usepackage{amsmath}
\usepackage{pdfpages}

\usepackage{tikz}
\usetikzlibrary{quantikz2}

\usepackage{float}
\floatstyle{boxed}
\newfloat{boxedalgorithm}{tbp}{loa}
\floatname{boxedalgorithm}{Algorithm}

\newcommand{\eq}[1]{\hyperref[eq:#1]{(\ref*{eq:#1})}}
\renewcommand{\sec}[1]{\hyperref[sec:#1]{Section~\ref*{sec:#1}}}
\newcommand{\thm}[1]{\hyperref[thm:#1]{Theorem~\ref*{thm:#1}}}
\newcommand{\lem}[1]{\hyperref[lem:#1]{Lemma~\ref*{lem:#1}}}
\newcommand{\cor}[1]{\hyperref[cor:#1]{Corollary~\ref*{cor:#1}}}
\newcommand{\itm}[1]{\hyperref[itm:#1]{\ref*{itm:#1}}}
\newcommand{\app}[1]{\hyperref[app:#1]{Appendix~\ref*{app:#1}}}
\newcommand{\dfn}[1]{\hyperref[dfn:#1]{Definition~\ref*{dfn:#1}}}
\newcommand{\fig}[1]{\hyperref[fig:#1]{Figure~\ref*{fig:#1}}}
\newcommand{\clm}[1]{\hyperref[clm:#1]{Claim~\ref*{clm:#1}}}
\newcommand{\alg}[1]{\hyperref[alg:#1]{Algorithm~\ref*{alg:#1}}}
\newcommand{\stp}[1]{\hyperref[stp:#1]{Step~\ref*{stp:#1}}}
\newcommand{\asm}[1]{\hyperref[asm:#1]{Assumption~\ref*{asm:#1}}}
\newcommand{\prot}[1]{\hyperref[prot:#1]{Protocol~\ref*{prot:#1}}}
\newcommand{\prob}[1]{\hyperref[prob:#1]{Problem~\ref*{prob:#1}}}
\newcommand{\rmk}[1]{\hyperref[rmk:#1]{Remark~\ref*{rmk:#1}}}
\newcommand{\tabl}[1]{\hyperref[tabl:#1]{Table~\ref*{tabl:#1}}}
\newcommand{\prop}[1]{\hyperref[prop:#1]{Proposition~\ref*{prop:#1}}}

\theoremstyle{definition}
\newtheorem{theorem}{Theorem}[section]
\newtheorem{lemma}[theorem]{Lemma}
\newtheorem{definition}[theorem]{Definition}

\newtheorem{corollary}[theorem]{Corollary}

\newcounter{prob}
\newtheorem{problem}[prob]{Problem}

\renewcommand{\ket}[1]{\vert #1\rangle}
\renewcommand{\bra}[1]{\langle #1\vert}

\DeclareMathOperator{\tr}{tr}
\DeclareMathOperator*{\Exp}{\mathbb{E}}
\renewcommand{\epsilon}[0]{\varepsilon}
\renewcommand{\i}[0]{\mathrm{i}}
\newcommand{\wt}[0]{\mathrm{wt}}
\newcommand{\bit}[0]{\{0,1\}}
\newcommand{\Binomial}[0]{\mathrm{Binomial}}

\title{Efficient Quantum State Identity Testing}
\author{Shih-Han Hung\footnote{Department of Electrical Engineering and Center for Quantum Science and Engineering, National Taiwan University. Email: shihhanh@ntu.edu.tw} 
\and Ming-Hsien Tsai\footnote{Department of Electrical Engineering, National Taiwan University. Email: minghsien@ntu.edu.tw}}
\date{}

\begin{document} 

\maketitle

\begin{abstract}
    We study the following quantum state identity problem: given access to copies of $n$ unknown pure quantum states, the goal is to determine whether every pair has fidelity at least $1-\varepsilon$ or some pair has fidelity at most $\varepsilon$. This problem relaxes the identical-or-orthogonal promise considered in previous work. For any fixed overlap parameter $\varepsilon\in(0,1/4)$, we present two quantum algorithms with complementary guarantees on circuit depth and sample complexity. The first runs in constant depth using parallel swap tests and takes $O(\log(n/\delta)\log(2/\delta))$ copies of each input state to succeed with probability at least $1-\delta$. The second uses the Schur transform to reduce the sample complexity to $O(\log(n/\delta))$, through an analysis relating Jucys--Murphy elements to Schur sampling. We also prove a matching sample complexity lower bound of $\Omega(\log(n/\delta))$ for any quantum algorithm with failure probability at most $\delta\le 1/2-c$, for any constant $c>0$, establishing that the second algorithm achieves the optimal dependence on both $n$ and $\delta$ for fixed $\varepsilon$.
\end{abstract}

\section{Introduction}

The \textit{quantum state identity problem} asks whether a collection of $n$ unknown quantum states all represent the same state.  A standard formulation assumes an exact promise: every pair of input states is either identical or orthogonal, and the task is to decide whether all of them are identical.  One may ask what remains possible when this promise is relaxed.  In particular, in the first case the states may only be close to one another, while in the second case some pair may have small, but nonzero, fidelity.

When the promise is exact, the standard quantum state identity problem is well studied. Kada, Nishimura, and Yamakami \cite{Kada_2008} showed that the \textit{Permutation Test} is optimal under the one-sided error requirement for arbitrary $n$. They also showed that, when $n$ is prime, the \textit{Circle Test} achieves the same optimal soundness with a much simpler circuit, since it only requires a Fourier transform over the cyclic group $C_n$, on a control register of dimension $n$, instead of a Fourier transform over the symmetric group $S_n$, on a control register of dimension $n!$ \cite{buhrman2026permutationtestsquantumstate}. Buhrman, Grinko, Verduyn Lunel, and Weggemans \cite{buhrman2026permutationtestsquantumstate} showed that the Permutation Test remains optimal when the one-sided error requirement is relaxed, in the sense that it maximizes the average success probability when the yes and no cases are equally likely. They also proposed, for the special case $n=2^m$, a simpler protocol based on an iterated wreath product of cyclic groups that uses only $n-1$ SWAP tests to approximate the Permutation Test.

These results rely heavily on the exact identical-or-orthogonal promise. The exact promise makes the underlying symmetry particularly transparent, but it is also quite restrictive. If several devices are intended to prepare the same state, small preparation errors will generally produce states that are close but not identical. Likewise, two states intended to be distinguishable need not be exactly orthogonal, and may have a small but nonzero overlap. This motivates formulating a robust version of quantum state identity in terms of quantitative bounds on pairwise fidelities.

A first version of this question already appeared in the closing discussion of \cite{Kada_2008}, where the authors retain perfect identity in the yes case but relax the no case to the existence of a pair with overlap at most $\varepsilon$. More recently, Buhrman et al. \cite{buhrman2026permutationtestsquantumstate} explicitly proposed an approximate promise on both sides of the problem: either the overlaps of all pairs are close to one, or at least one pair has overlap close to zero.

We write $\mathrm{QSI}_{n}^\varepsilon$ for the $\varepsilon$-approximate quantum state identity problem on $n$ pure quantum states; the formal definition of our relaxed problem is given below. Our use of fidelity differs slightly from the overlap formulation in \cite{buhrman2026permutationtestsquantumstate}, but captures the same qualitative question: can identity still be tested efficiently when the states are only approximately identical or approximately orthogonal?

\begin{definition}[$\mathrm{QSI}_{n}^\varepsilon$ under two-sided relaxation]\label{dfn:two-sided QSI}
    Given a collection of $n\ge 2$ pure quantum states $\ket{\psi_1}, \dots, \ket{\psi_n}$ as inputs, each of local dimension $d\ge 2$, and an overlap parameter $\varepsilon\in[0,1/2)$, we consider the promise problem defined by the following conditions:
    \begin{itemize}
        \item \textbf{Yes instance:} For all pairs of distinct indices $i, j \in [n]$, it holds that $|\braket{\psi_i}{\psi_j}|^2 \ge 1 - \varepsilon$.
        \item \textbf{No instance:} There exists at least one pair of distinct indices $a, b \in [n]$ such that $|\braket{\psi_a}{\psi_b}|^2 \le \varepsilon$.
    \end{itemize}
\end{definition}

The restriction $\varepsilon<1/2$ ensures that the yes and no promises are disjoint, and we assume $n\ge 2$ and $d\ge 2$ throughout, since otherwise there are no no instances.
The output of any algorithm for $\mathrm{QSI}_{n}^\varepsilon$ is either $\mathtt{accept}$ or $\mathtt{reject}$.
We say that the algorithm has completeness error at most $\delta$ if, for every yes instance, it outputs \texttt{accept} with probability at least $1-\delta$,
and that it has soundness error at most $\delta$ if, for every no instance, it outputs \texttt{reject} with probability at least $1-\delta$.
An algorithm solves $\mathrm{QSI}_{n}^\varepsilon$ with error $\delta$ if both its completeness and soundness errors are at most $\delta$.

Under the additional promise that every pair of input states is either identical or orthogonal, the Permutation Test uses one copy of each input state, achieves perfect completeness and soundness error at most $1/n$, and is optimal among all one-copy measurements, both under the one-sided error requirement \cite{Kada_2008} and with respect to the average success probability when the yes and no cases are equally likely \cite{buhrman2026permutationtestsquantumstate}. At $\epsilon=0$, \dfn{two-sided QSI} only requires a single orthogonal pair in a no instance, so this soundness guarantee does not apply to the problem defined here: for instance, for $n=4$ the no instance $\ket{0},\ket{1},\ket{+},\ket{+}$ is accepted by the Permutation Test with probability $7/24>1/4$. Moreover, the perfect
completeness of the Permutation Test is not robust under our two-sided relaxation $\mathrm{QSI}_n^\epsilon$. To see this, let $n=2m$, set
$\ket{\psi_1}=\cdots=\ket{\psi_m}=\ket{0}$, and let
$\ket{\psi_{m+1}}=\cdots=\ket{\psi_{2m}}=\ket{\phi}$, where
$\ket{\phi}=\sqrt{1-\epsilon}\ket{0}+\sqrt{\epsilon}\ket{1}$.
This is a valid yes instance, since every pair has fidelity at least
$1-\epsilon$. Nevertheless, a direct calculation shows that the
acceptance probability of the permutation test is
\begin{align}
p_{\mathrm{accept}}
&\leq(n+1)\left(\frac{1+\sqrt{1-\epsilon}}{2}\right)^n.
\end{align}
Indeed, for qubits the projector onto the symmetric subspace can be written as $\Pi_{\mathrm{sym}}=(n+1)\int\proj{\psi}^{\otimes n}\,d\psi$, where the integral is over the Haar measure on pure qubit states, so that $p_{\mathrm{accept}}=\bra{\Psi}\Pi_{\mathrm{sym}}\ket{\Psi}\le(n+1)\max_{\psi}\left(|\braket{\psi}{0}|\,|\braket{\psi}{\phi}|\right)^{n}$ for $\ket{\Psi}=\ket{0}^{\otimes m}\otimes\ket{\phi}^{\otimes m}$; the bound then follows from $|\braket{\psi}{0}|\,|\braket{\psi}{\phi}|\le\frac12\left(|\braket{\psi}{0}|^2+|\braket{\psi}{\phi}|^2\right)\le\frac{1+|\braket{0}{\phi}|}{2}$, where the last step uses that the nonzero eigenvalues of $\proj{0}+\proj{\phi}$ are $1\pm|\braket{0}{\phi}|$.
For every fixed $\epsilon>0$, this probability is exponentially small in
$n$. Hence, the permutation test may reject a valid yes instance with
probability approaching one.

This example rules out a direct extension of the one-copy permutation
test to the relaxed problem. We therefore allow the algorithm to use
multiple copies of each input state and ask how many copies are needed.
We use the number of copies required of each input state as our primary resource measure, while also keeping track of the gate count and depth of the circuits.

The remaining difficulty is locating a bad pair. A no instance may
contain only one low-fidelity pair, whose indices are unknown. Testing
all pairs separately would use too many copies, whereas a uniformly
chosen pair hits the bad pair with probability only $\Theta(1/n^2)$. Our algorithms avoid this direct search by exploiting additional structure in the input.

\subsection{Our results}

In this paper, we present two algorithms with complementary guarantees on sample complexity and circuit depth, which resolve the open problem raised in \cite{buhrman2026permutationtestsquantumstate}.
Given $n$ pure quantum states and an overlap parameter $\varepsilon\in(0,1/4)$, the goal is to distinguish between two cases: either every pair of states has fidelity at least $1-\varepsilon$, or there exists a pair with fidelity at most $\varepsilon$.

\subsubsection{Algorithms for quantum state identity}

Our first result is a constant-depth algorithm that uses parallel swap tests applied to random matchings on the input states. 

\begin{theorem}[Informal statement of \thm{randommatching-soundness} and
\thm{randommatching-completeness}]\label{thm:informal-random-matching}
For every $\epsilon\in(0,1/4)$ and $\delta\in(0,1)$, let
$\gamma=\frac12-\sqrt{\epsilon}$. There exists a randomized
constant-depth quantum algorithm that solves
$\mathrm{QSI}_n^\epsilon$ with error probability at most $\delta$
using
\begin{align}\label{eq:informal-swap}
O\left(
\frac{1}{\gamma^2}
\log\left(\frac{n}{\delta}\right)
\log\left(\frac{2}{\delta}\right)
\right)
\end{align}
copies of each input state.
\end{theorem}

Detecting a low-fidelity pair in a no instance can be difficult, since we are only guaranteed that one such pair exists.
Performing swap tests on a random matching consisting of $n/2$ disjoint pairs captures this particular pair with probability only $O(1/n)$.
To deal with this difficulty, we show that the existence of one such pair implies the existence of $\Omega(n)$ pairs with fidelity at most $1-\epsilon-\Omega(\gamma)$.
As a result, a random matching contains a detectable low-fidelity pair with probability $\Omega(1)$.
With this observation, we sample $O(\log(2/\delta))$ independent random matchings to amplify the probability of pairing two quantum states with sufficiently low fidelity to $1-O(\delta)$.

The above steps effectively reduce $\mathrm{QSI}_n^\epsilon$ to deciding whether the fidelity of a pair of quantum states is at least $1-\epsilon$ or at most $1-\epsilon-\Omega(\gamma)$.
Conditioned on the event that such a pair has been sampled, a single swap test distinguishes the two cases with bias $\Omega(\gamma)$.
Using $O(\frac{1}{\gamma^2}\log(n/\delta))$ repetitions allows us to amplify the success probability to $1-O(\delta/n)$.
By a union bound, we conclude that the two-sided error can be controlled to within $\delta$ using the number of copies in \eq{informal-swap}.

While the algorithm establishing \thm{informal-random-matching} is very simple, it incurs an extra $O(\log(2/\delta))$ multiplicative factor in the sample complexity.
To further improve the sample complexity, we present a sample-optimal algorithm based on the Schur transform and Jucys--Murphy operators.

\begin{theorem}[Informal statement of \thm{jucysmurphy-soundness} and
\thm{jucysmurphy-completeness}]\label{thm:informal-schur}
For every $\epsilon\in(0,1/4)$ and $\delta\in(0,1)$, let
$\gamma=\frac12-\sqrt{\epsilon}$. There exists a quantum
algorithm that solves $\mathrm{QSI}_{n}^\varepsilon$ with error probability
at most $\delta$ using
\begin{align}
O\left(
\frac{1}{\gamma^2}
\log\left(\frac{n}{\delta}\right)
\right)
\end{align}
copies of each input state.
\end{theorem}

The algorithm uses the Jucys--Murphy operators to aggregate pairwise fidelity information among the input states.
In more detail, the Jucys--Murphy operator $J_i$ for the $i$th state is the sum of all transpositions $(j\,i)$ with $j<i$.
As an observable on the input states, its expectation $\langle J_i\rangle$ equals the sum of the fidelities between state $i$ and the states $j<i$.
Therefore, estimating the expectation of $J_i$ for both the ascending and the descending ordering of the inputs gives the sum of the fidelities between $i$ and all $j\neq i$.
The promise of $\mathrm{QSI}_n^\epsilon$ implies that in a yes instance this sum is at least $(n-1)(1-\epsilon)$ for all $i\in[n]$, whereas in a no instance there must be some $i$ for which the sum is at most $(n-1)(1-\epsilon-\Omega(\gamma))$.

To estimate $\langle J_i\rangle$, we use the Schur transform to measure the states in the Schur--Weyl basis $\{\ket{\lambda,S,T}\}$, where $\lambda$ is a partition of $n$, $S$ is a standard Young tableau of shape $\lambda$, and $T$ is a semistandard Young tableau of shape $\lambda$.
From an outcome $(\lambda,S)$ one obtains an unbiased sample $\mathrm{ct}_S(i)$ of each $\langle J_i\rangle$, since $\langle J_i\rangle = \Exp_{\lambda,S}[\mathrm{ct}_{S}(i)]$, where $\mathrm{ct}_S(i)$ is the content of the box containing $i$ in the standard Young tableau $S$.
Since the $J_i$ commute with one another, the same outcome can be used to compute the content of every $i$ simultaneously.
Therefore, performing the Schur transform on two blocks of input states (one in each ordering) yields $n$ samples $z_1,\ldots,z_n$, where each $z_i$ has expectation equal to the sum of the fidelities between $i$ and all $j\neq i$.
Repeating the process $O(\frac{1}{\gamma^2}\log(n/\delta))$ times allows us to estimate this quantity for each state $i$ to precision $O(\gamma n)$ with probability $1-O(\delta/n)$, which suffices for a final union bound.

In addition to the sample complexity, we also analyze the gate count and circuit depth of these algorithms, summarized in \tabl{qsi-comparison}.

\begin{table}[h!]
\centering
\small
\renewcommand{\arraystretch}{1.3}
\setlength{\tabcolsep}{4pt}

\begin{tabular}{
    >{\raggedright\arraybackslash}p{3.1cm}
    >{\raggedright\arraybackslash}p{1.5cm}
    >{\raggedright\arraybackslash}p{3.0cm}
    >{\raggedright\arraybackslash}p{4.5cm}
    >{\raggedright\arraybackslash}p{3.2cm}
}
\hline
Algorithm
& Setting
& Copies per input state
& Gate count
& Circuit depth \\
\hline

Permutation Test \cite{Kada_2008,buhrman2026permutationtestsquantumstate}
&
Exact
&
1
&
$O(n^2\log d)$ \cite{10.1098/rspa.2022.0733}
&
$O(n^2\log d)$ \cite{10.1098/rspa.2022.0733}
\\

Iterated Swap Tree
\cite{buhrman2026permutationtestsquantumstate}
&
Exact, $n=2^m$
&
1
&
$O(n\log n\log d)$
&
$O(\log n\,\log(n\log d))$
\\

Random matching
(this work)
&
Relaxed
&
$\displaystyle
O\left(
\frac{1}{\gamma^2}
\log\frac{n}{\delta}
\log\frac{2}{\delta}
\right)$
&
$\displaystyle
O\left(
\frac{n\log d}{\gamma^2}
\log\frac{n}{\delta}
\log\frac{2}{\delta}
\right)$
&
$O(1)$
\\

Jucys--Murphy protocol (this work)
&
Relaxed
&
$\displaystyle
O\left(
\frac{1}{\gamma^2}
\log\frac{n}{\delta}
\right)$
&
$\displaystyle
\tilde O\left(
\frac{1}{\gamma^2}
\log\frac{n}{\delta}
\cdot\min\{n^4,nd^4\}
\right)$
&
$\displaystyle
\tilde O\left(
\min\{n^4,nd^4\}
\right)$
\\

\hline
\end{tabular}

\caption{
Comparison of QSI algorithms for $n$ input states of local dimension $d$, where
$\gamma=\frac12-\sqrt{\varepsilon}$. ``Exact'' refers to the promise of \cite{Kada_2008,buhrman2026permutationtestsquantumstate} that every pair of input states is either identical or orthogonal, and ``Relaxed'' to \dfn{two-sided QSI}.
Gate counts and depths are in terms of elementary gates on qubits, in the cost model described below; $\tilde O$ hides polylogarithmic factors in $n$, $d$, and the inverse precision of the Schur transform.
The bounds for our algorithms assume parallel execution on
disjoint copies.
}
\label{tabl:qsi-comparison}
\end{table}
\paragraph{Cost model for \tabl{qsi-comparison}.}
Each $d$-dimensional input state is encoded in $q:=\lceil\log_2 d\rceil$ qubits, and gate counts and depths are measured in one-, two-, and three-qubit gates (in particular Fredkin gates), with ancilla qubits allowed.
The random-matching protocol only uses the outcome statistics of its swap tests, so each swap test can be implemented destructively by measuring the $q$ pairs of corresponding qubits of the two registers in the Bell basis and computing the parity of the number of singlet outcomes (see \sec{swap-test}); this uses $O(q)$ gates on disjoint pairs of qubits and has depth $O(1)$, which gives the gate count $O(nk\ell\log d)$ and the constant depth listed in the table.
The Iterated Swap Tree, in contrast, needs the coherent swap test, since the swap tests at later levels act on the post-measurement states of earlier levels. A coherent swap test between two registers of $q'$ qubits each uses $O(q')$ gates: a Hadamard gate on a control qubit, $q'$ Fredkin gates sharing this control, and a second Hadamard gate.
Since the Fredkin gates share a single control qubit, applying them sequentially gives depth $O(q')$; alternatively, the control qubit can be fanned out to $q'$ ancilla qubits by a tree of CNOT gates of depth $\lceil\log_2 q'\rceil$, after which the Fredkin gates can be applied in parallel and the fan-out uncomputed, giving depth $O(\log q')$.
The Iterated Swap Tree performs $n-1$ such swap tests arranged in $\log_2 n$ levels, where the swap tests at level $l\in\{0,\ldots,\log_2 n-1\}$ act on blocks of $2^l$ registers, i.e., $q'=2^l q$; summing over the levels gives $O(n\log n\log d)$ gates and depth $O(\sum_l \log(2^lq))=O(\log^2 n+\log n\log\log d)=O(\log n\,\log(n\log d))$.
For the Permutation Test, we use the construction of Bradshaw, LaBorde, and Wilde \cite{10.1098/rspa.2022.0733}, which implements the projection onto the symmetric subspace with $O(n^2)$ controlled register swaps, i.e., $O(n^2\log d)$ Fredkin gates; its depth was not analyzed, but it is trivially bounded by the gate count. The alternative implementation through the quantum Fourier transform over the symmetric group has gate count and depth $\tilde O(n^3)$ \cite{bruinsma2026quantumfouriertransformsymmetric}, independently of $d$ up to the cost of the controlled permutation.
For the Jucys--Murphy protocol, the cost is dominated by the Schur transform; see \sec{optimal}.

\subsubsection{Sample complexity lower bound}

We prove a matching lower bound on the sample complexity of this problem.
This lower bound shows that our upper bounds are tight in the regime where the overlap parameter is a constant.

\begin{theorem}[Informal statement of \thm{lower-bound}]
For every $\epsilon\in(0,1/4)$, every constant $c>0$, and every $0<\delta\leq 1/2-c$, any quantum
algorithm that solves $\mathrm{QSI}_{n}^\varepsilon$ with error probability
at most $\delta$ requires
\begin{align}\label{eq:informal-lower-bound}
\Omega\left(
\frac{\log(n/\delta)}{\log(1/\epsilon)}
\right)
\end{align}
copies of each input state.
\end{theorem}

For constant $\epsilon$, this lower bound is
$\Omega(\log(n/\delta))$, which matches the sample complexity of the algorithm based on the Schur transform up to a constant factor. When both $\epsilon$ and $\delta$ are constants, the lower bound becomes $\Omega(\log n)$,
and both of our algorithms use $O(\log n)$ copies of each input state.

We establish the lower bound by showing that the following hard instance requires the number of copies in \eq{informal-lower-bound}.
Each input is a random single-qubit state near either the north pole or the south pole of the Bloch sphere, with a uniformly random azimuthal angle and with the polar angle fixed to be $\arccos\sqrt{1-\epsilon}$ away from the pole.
If all of the states are near the north pole, then we obtain a yes instance; if one of them is near the south pole, we obtain a no instance.
Any algorithm solving the problem must in particular distinguish a random yes instance from a random no instance in which the state $\ket{\psi_i}$ is near the south pole for some location $i$.
When $i$ is fixed, this reduces to a state discrimination problem between two fixed pure states whose fidelity is exactly $\epsilon$, which requires $\Omega(\log_{1/\epsilon}(1/\delta))$ copies.

To prove a lower bound that depends on $n$, we consider the case where $\ket{\psi_i}$ is near the south pole for a \emph{random} $i$.
We show that, after averaging over the azimuthal angle, $m$ copies of the state in the $i$th register are described by a binomial distribution $\mathrm{Binomial}(m,p)$ with $m$ trials and success probability $p$ in the Dicke basis.
When $\ket{\psi_i}$ is near the north pole, we have $p=\frac{1-\sqrt{1-\epsilon}}{2}$; otherwise, we have $p=\frac{1+\sqrt{1-\epsilon}}{2}$.
Thus the problem reduces to deciding whether, among $n$ independent binomial random variables, one of them (at a random location) has the complementary parameter.
We then prove that the total variation distance between the two cases of this classical problem is at most $O((4/\epsilon)^{m/2}/\sqrt{n})$.
This implies that distinguishing the two cases with constant success probability requires $m=\Omega(\log_{1/\epsilon}(n))$ copies.

\subsection{Discussion}
The relaxed problem behaves differently from exact quantum state identity. Under the exact promise, the Permutation Test uses one copy of each input state and has perfect completeness. This guarantee can fail sharply under the relaxed promise. Even when every pair has fidelity at least $1-\varepsilon$, the input may have exponentially small overlap with the symmetric subspace. The one-copy measurement
that is optimal in the exact setting is therefore not reliable in the relaxed setting, and our lower bound shows that multiple copies are
necessary.

Our two algorithms approach this issue differently. Our first algorithm uses only parallel swap tests and can be implemented in constant depth.
Its analysis uses the fact that a low-fidelity pair cannot be isolated: it must induce many additional witness pairs that can also be detected. This simpler circuit comes with an additional $\log(2/\delta)$ factor in the sample complexity.
Our second algorithm does not attempt to locate a particular low-fidelity pair. Instead, measurements in the Schur basis, which estimate the Jucys--Murphy observables on the forward and reverse orderings of the inputs, provide an estimator for the total fidelity involving each state.
This gives a sample complexity of $O(\gamma^{-2}\log(n/\delta))$, where $\gamma=\frac12-\sqrt{\varepsilon}$.
For fixed $\varepsilon$, this matches our lower bound up to constant factors. The tradeoff is that the algorithm requires Schur transforms and therefore has a much larger gate complexity.
We believe that this technique for fidelity estimation may be useful in other applications.

One limitation of our algorithms is that they only work in the regime $0<\epsilon<1/4$. This restriction is needed to separate the high-fidelity pairs in a yes instance from the witness pairs guaranteed by our geometric argument.
Since this restriction is only needed for our analysis, we do not believe that it is inherent to the problem itself.
Whether comparable algorithms exist in the full range $\varepsilon<1/2$ remains open. Our lower bound is tight for fixed $\varepsilon$, but it leaves open the optimal dependence on $\varepsilon$, both near the current threshold $1/4$ and in the larger range up to $1/2$.

We have restricted our attention to pure input states. In general, a swap test acting on two states $\rho,\sigma$ estimates $\tr(\rho\sigma)$, which is not equal to the fidelity unless one of the states is pure. Thus algorithms based on applying transpositions to the states do not seem to work immediately.
Extending our algorithms to mixed inputs would therefore seem to require a different idea.

\section{Preliminaries}
Throughout the paper, $\log$ denotes the natural logarithm, and $[n]:=\{1,\ldots,n\}$. In the following subsections, we review the combinatorial and representation-theoretic background needed for our algorithms. For a more detailed introduction to these topics, see \cite{wright2016learn}.

\subsection{Partitions and Young diagrams}

A partition of $n$, denoted by $\lambda\vdash n$, is a tuple of integers $\lambda=(\lambda_1,\ldots, \lambda_d)$ satisfying $\lambda_1 + \cdots+\lambda_d=n$ and $\lambda_1\ge\cdots\ge \lambda_d \ge 0.$ The length $\ell(\lambda)$ is the number of nonzero entries of $\lambda$. Partitions $\lambda$ are usually represented pictorially as \textit{Young diagrams}, which contain $\lambda_r$ boxes in row $r$. A \textit{standard Young tableau (SYT)} $S$ of shape $\lambda$ is a filling of the boxes with the numbers in $[n]$, each appearing exactly once, such that the entries strictly increase from left to right along each row and from top to bottom along each column. We write
$\operatorname{SYT}(\lambda)$ for the set of standard Young tableaux of shape $\lambda$. A \textit{semistandard Young tableau (SSYT)} $T$ of shape $\lambda$ with alphabet $[d]$ is a filling of the boxes with numbers in $[d]$ such that the entries are weakly increasing from left to right along each row and strictly increasing from top to bottom along each column. We illustrate these concepts in \fig{young-tableaux}.

\begin{figure}[h!]
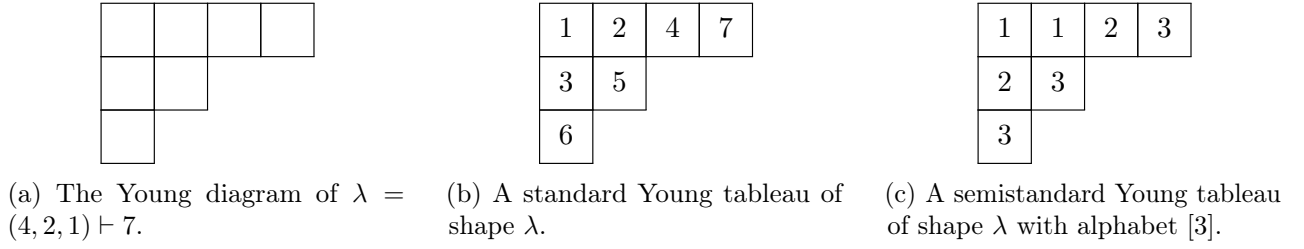

    
    \centering

    \begin{subfigure}[t]{0.31\textwidth}
        \centering
        \ydiagram{4,2,1}
        \caption{The Young diagram of $\lambda=(4,2,1)\vdash 7$.}
        \label{fig:young-diagram}
    \end{subfigure}
    \hfill
    \begin{subfigure}[t]{0.31\textwidth}
        \centering
        \begin{ytableau}
            1 & 2 & 4 & 7 \\
            3 & 5 \\
            6
        \end{ytableau}
        \caption{A standard Young tableau
        of shape $\lambda$.}
        \label{fig:syt}
    \end{subfigure}
    \hfill
    \begin{subfigure}[t]{0.31\textwidth}
        \centering
        \begin{ytableau}
            1 & 1 & 2 & 3 \\
            2 & 3 \\
            3
        \end{ytableau}
        \caption{A semistandard Young tableau
        of shape $\lambda$ with alphabet $[3]$.}
        \label{fig:ssyt}
    \end{subfigure}
    \caption{Examples of a Young diagram and Young tableaux.}
    \label{fig:young-tableaux}
\end{figure}

\begin{definition}[Content]\label{dfn:content}
    For a box in row $r$ and column $c$ of a Young diagram, define its
    content by
    \begin{align}
         \operatorname{ct}(r,c):=c-r.
    \end{align}
   
    For a standard Young tableau $S$, we denote by
    $\operatorname{ct}_S(i)$ the content of the box containing $i$.
\end{definition}

\subsection{Representation theory of symmetric and general linear groups}

The symmetric group $S_n$ is the group of all permutations of $n$
elements. Its irreducible representations are indexed by partitions $\lambda\vdash n$ and are written $(\kappa_\lambda,\mathrm{Sp}_\lambda)$, where $\mathrm{Sp}_\lambda$ is the \textit{Specht module}. We choose a particular basis called \textit{Young's orthogonal basis}, which is an orthonormal basis for $\mathrm{Sp}_\lambda$ consisting of one vector $\ket{S}$ for each standard Young tableau $S$ of shape $\lambda$. In this basis, the matrices $\kappa_\lambda(\pi)$ are real and orthogonal for every $\pi\in S_n$. We write $\dim(\lambda):=\dim(\mathrm{Sp}_\lambda)$ for the dimension of the Specht module, which is equal to the number of standard Young tableaux of shape $\lambda$.

The general linear group $GL(d)$ is the group of invertible
$d\times d$ complex matrices. Its polynomial irreducible representations are indexed by partitions $\lambda\vdash n$ satisfying $\ell(\lambda)\leq d$ and are written $(\nu_\lambda^d,V_\lambda^d)$, where $V_\lambda^d$ is the \textit{Weyl module}. The term ``polynomial'' means that the entries of $\nu_\lambda^d(M)$ are polynomials in the matrix entries of $M$. There is an orthonormal basis for $V_\lambda^d$ called \textit{Gelfand--Tsetlin basis}, whose basis vectors $\ket{T}$ are indexed by semistandard Young tableaux $T$ of shape $\lambda$ with alphabet $[d]$. Restricting $\nu_\lambda^d$ to the unitary group $U(d)$ also gives an irreducible representation.

There is a natural representation $P(\cdot)$ of $S_n$ acting on the space $(\mathbb{C}^d)^{\otimes n}$ as follows: for each $\pi\in S_n$, 
\begin{align}
P(\pi)\ket{i_1,\ldots,i_n}=\ket{i_{\pi^{-1}(1)},\ldots,i_{\pi^{-1}(n)}}.
\end{align}

Similarly, there is also a natural representation $Q^d(\cdot)$ of $GL(d)$ acting on this space as follows: for each $M\in GL(d)$, 
\begin{align}
Q^d(M)\ket{i_1,\ldots,i_n}=(M\ket{i_1})\otimes\cdots\otimes(M\ket{i_n}).
\end{align}

These two representations commute, that is,
$P(\pi)Q^d(M)=Q^d(M)P(\pi)$ for every $\pi\in S_n$ and
$M\in GL(d)$. Their joint action on
$(\mathbb{C}^d)^{\otimes n}$ can be simultaneously block-diagonalized into irreducible representations. This fact is known as \textit{Schur--Weyl duality}, and it states that there exists a unitary $U^d_{\mathrm{Schur}}$, called the \textit{Schur transform}, such that for every $\pi\in S_n$ and $M\in GL(d)$,
\begin{align}
U_{\mathrm{Schur}}^dP(\pi)Q^d(M)(U^d_{\mathrm{Schur}})^\dagger
=\sum_{\substack{\lambda\vdash n\\\ell(\lambda)\leq d}}\ket{\lambda}\bra{\lambda} \otimes
\kappa_\lambda(\pi)\otimes\nu_\lambda^d(M).
\end{align}

Here, we write $U_{\mathrm{Schur}}^d$ to make the dependence on the dimension $d$ explicit. In our algorithm, however, $d$ is determined by the input, so we write $U_{\mathrm{Schur}}$ for simplicity. The Schur transform changes the basis from the standard basis to the \textit{Schur basis}, which consists of the vectors $\ket{\lambda, S, T}$, where $\lambda\vdash n$ is a partition with $\ell(\lambda)\le d$, $S\in\operatorname{SYT}(\lambda)$, and $T$ is a semistandard Young tableau of shape $\lambda$ with alphabet $[d]$. An efficient algorithm for computing the Schur transform was given by Bacon, Chuang, and Harrow \cite{baconQuantumSchurClebschGordan2007}. It is shown in the footnote on page 160 of Harrow's Ph.D. thesis \cite{harrow2005applicationscoherentclassicalcommunication} that it can be computed to accuracy $\xi$ in diamond distance in time $\operatorname{poly}(n, \log(d), \log(1/\xi))$. Pelecanos, Spilecki, and Wright \cite{pelecanos2025debiasedkeylsalgorithmnew} showed that this algorithm can produce Young's orthogonal basis on the symmetric-group register.

\begin{definition}[Jucys--Murphy element]\label{dfn:jucysmurphy}
    For $i\in[n]$, the $i$-th Jucys--Murphy element of the group algebra
    $\mathbb{C}[S_n]$ is defined by 
    \begin{align}
        J_1:=0, \qquad J_i:=\sum_{j<i}(j\,i), \qquad i=2,\ldots,n,
    \end{align}
 
    where $(j\,i)$ denotes the transposition exchanging $j$ and $i$.
\end{definition}

In our quantum setting, we use the same notation $P$ for the linear
extension of the natural permutation representation,
$P:\mathbb{C}[S_n]\longrightarrow
\operatorname{End}((\mathbb{C}^d)^{\otimes n})$.
For a transposition $(j\,i)$, the operator $P_{ij}:=P((j\,i))$ is the SWAP
operator between the $j$-th and $i$-th tensor factors. Therefore, for
every $i\in[n]$, the represented Jucys--Murphy operator is
\begin{align}\label{jucys-fidelity}
    P(J_i)=\sum_{j<i}P_{ij}.
\end{align}

\begin{lemma}[{\cite[Theorem~2.36]{pelecanos2025debiasedkeylsalgorithmnew}}]
\label{lem:jucys-murphy-eigenvalue}
Let $\lambda\vdash n$. For every $i\in[n]$, Young’s orthogonal basis of $\mathrm{Sp}_\lambda$ is an eigenbasis of $\kappa_\lambda(J_i)$, where $\kappa_\lambda(J_i)$ is an operator on $\mathrm{Sp}_\lambda$. More precisely, for every $S\in\operatorname{SYT}(\lambda)$,
\begin{align}
    \kappa_\lambda(J_i)\ket{S}=\operatorname{ct}_S(i)\ket{S}.
\end{align}
\end{lemma}
Consequently, Young's orthogonal basis is a common eigenbasis of the
operators $\kappa_\lambda(J_1),\ldots,\kappa_\lambda(J_n)$.

We use the same notation $J_i$ for the group-algebra element and its
represented operator whenever the representation is clear from context.

\section{A constant-depth algorithm for quantum state identity}

In this section, we construct a quantum algorithm that uses only swap tests in parallel, together with classical pre- and post-processing, to solve the quantum state identity problem.
Compared to the lower bound we establish in \sec{lower-bound}, in the constant-threshold regime it incurs an extra multiplicative factor $O\left(\log(2/\delta)\right)$ from an amplification using independent random matchings.
To explain the ideas behind our algorithm, in \sec{swap-test} we consider the simplest case $n=2$; in this case, the two input states are either close to or far from each other.
A single swap test then distinguishes the two cases with a bias of $\Omega(\gamma)$ in its acceptance probability, where $\gamma=\frac12-\sqrt{\epsilon}$.
To amplify this bias, we run $k$ swap tests in parallel, which we call a $k$-fold swap test; for $k=O\left(\frac{1}{\gamma^2}\log(n/\delta)\right)$, the test decides a pair correctly with probability $1-O(\delta/n)$.

To extend this result to general $n\geq 3$, in \sec{random-matching} we run the $k$-fold swap test on a random matching that pairs the $n$ registers into $n/2$ pairs (if $n$ is odd, we first reduce to the even case by adding one extra copy of one of the input states as an additional register).
For a no instance, it is only guaranteed that there is an unknown pair of states that are far from each other in fidelity.
Thus, when sampling a random matching, we only obtain an edge whose endpoints are sufficiently far apart with constant probability.
To amplify the probability that such a pair is caught by a swap test, our algorithm samples $\ell=O\left(\log(2/\delta)\right)$ random matchings, so that two registers with small overlap are paired with probability $1-O(\delta)$.
On the other hand, for a yes instance, no matter how the perfect matchings are chosen, each pair must have fidelity at least $1-\epsilon$.
Since the number $\ell$ of random matchings is small, the probability that the $k$-fold swap test rejects any of the pairs can be bounded by $O(\delta)$.

Before presenting our randomized protocol for solving
$\mathrm{QSI}_{n}^\epsilon$ (\dfn{two-sided QSI}), we first establish
a geometric property of no instances. This property allows us to exploit
the existence of a low-fidelity pair in the subsequent random-matching
analysis.

\begin{lemma}[Geometric witness lemma]
\label{lem:geometric-witness}
Suppose that there exist distinct indices $a,b\in[n]$ such that
$F_{ab}:=|\langle\psi_a|\psi_b\rangle|^2\le \epsilon$.
Then, for every $r\in[n]\setminus\{a,b\}$,
\begin{align}
\min\{F_{ar},F_{br}\}
\le
\frac{1+\sqrt{\epsilon}}{2}.
\end{align}
\end{lemma}
\begin{proof}
    For pure states, let $\theta_{ij}=\arccos|\langle\psi_i|\psi_j\rangle|\in[0,\pi/2]$ be the Bures angle, which is a metric on pure states.
    By the triangle inequality,
    we have $\theta_{ab}\leq \theta_{ar} + \theta_{rb}$.
    This means that $\max\{\theta_{ar}, \theta_{rb}\}\geq \theta_{ab}/2$, i.e., at least one of $\theta_{ar}$ and $\theta_{rb}$ is at least $\theta_{ab}/2$.
    Without loss of generality, assume that $\theta_{ar}\geq \theta_{rb}$, which implies that $\theta_{ar}\geq\theta_{ab}/2$.
    Since $\cos\theta_{ab}=\sqrt{F_{ab}}\leq\sqrt{\epsilon}$ and $\cos$ is monotonically decreasing on $[0,\pi/2]$, we have
    \begin{align}
        F_{ar}=\cos^2\theta_{ar}\leq\cos^2(\theta_{ab}/2)=\frac{1+\cos\theta_{ab}}{2}\leq\frac{1+\sqrt\epsilon}{2}.
    \end{align}
    Moreover, $\theta_{ar}\geq\theta_{rb}$ implies $F_{ar}\leq F_{br}$, so $F_{ar}=\min\{F_{ar},F_{br}\}$, which completes the proof.
\end{proof}

\lem{geometric-witness} shows that a single bad pair cannot be
geometrically isolated from all of the remaining states. In particular,
if $(a,b)$ satisfies $F_{ab}\le\epsilon$, then for every
$r\in[n]\setminus\{a,b\}$, at least one of the two edges $(a,r)$ and
$(b,r)$ has fidelity at most
$\frac{1+\sqrt{\epsilon}}{2}$.
We refer to such additional pairs as \emph{witness edges}. Thus, although
the promise of a no instance only guarantees the existence of one bad
pair, the existence of this pair necessarily induces a collection of
witness edges incident to $a$ or $b$. This geometric structure allows
the random-matching protocol to search for a witness edge rather than
requiring it to sample the original bad pair directly.

In the following subsections, we first construct a $k$-fold swap test
that distinguishes a high-fidelity pair from a witness pair with high
probability. We then apply this pairwise test to independently sampled
random perfect matchings and show that the geometric structure above
allows a no instance to be detected with sufficiently large probability.
Repeating the procedure over
$\ell=O(\log(2/\delta))$ independent matchings amplifies the overall
success probability to at least $1-\delta$.

\subsection{A \texorpdfstring{$k$}{k}-fold swap test for \texorpdfstring{$n=2$}{n=2}}\label{sec:swap-test}

We first consider the case $n=2$: given copies of $\ket{\psi}$ and $\ket{\phi}$, determine whether $|\braket{\psi}{\phi}|^2\geq 1-\epsilon$ or $|\braket{\psi}{\phi}|^2\leq \frac{1+\sqrt\epsilon}{2}$.
This is a harder problem than $\mathrm{QSI}_2^\epsilon$, since the threshold for the no instance is larger than $\epsilon$; we use this stronger guarantee so that the same test can
later be applied to the witness edges guaranteed by
\lem{geometric-witness}.
For this problem to make sense, we require
$1-\epsilon > \frac{1+\sqrt\epsilon}{2}$,
which is equivalent to $\epsilon< 1/4$.
This is the source of the restriction $\epsilon<1/4$ in our results.

For this special case, we use the $k$-fold swap test, which takes $k$ copies of each state and runs a swap test on each pair of copies (see the full description in \alg{k-fold-swap-test}).
The $k$-fold swap test accepts if at least $k\tau$ of the swap tests accept, where $\tau=\frac{7-2\epsilon+\sqrt{\epsilon}}{8}$ is the threshold, and rejects otherwise.

Recall that the standard swap test on $\ket{\psi}\ket{\phi}$ applies a Hadamard gate to a control qubit, a controlled-SWAP of the two registers, and a second Hadamard gate, and then measures the control qubit; the outcome $0$, which we call the \emph{symmetric outcome}, occurs with probability $\frac12+\frac12|\braket{\psi}{\phi}|^2$.
Our analysis only uses this outcome distribution, so any implementation that reproduces it can be used. 
In particular, we implement the swap tests \emph{destructively}, without a control qubit \cite{garciaescartin2013swap,cincio2018learning}: writing each $d$-dimensional register as $q=\lceil\log_2 d\rceil$ qubits, we measure each of the $q$ pairs of corresponding qubits of the two registers in the Bell basis (a CNOT gate followed by a Hadamard gate on the first qubit of the pair, followed by a computational-basis measurement of both qubits), and declare the symmetric outcome if and only if the number of pairs that yield the singlet $\ket{\Psi^-}=\frac{1}{\sqrt2}(\ket{01}-\ket{10})$ (i.e., the measurement outcome $11$) is even.
This reproduces the outcome distribution of the standard swap test: since the two-qubit SWAP operator equals $I-2\proj{\Psi^-}$, the swap operator on the two registers equals $\bigotimes_{t=1}^{q}(I-2\Pi^-_t)$, where $\Pi^-_t$ is the singlet projector on the $t$-th pair of qubits, and expanding this product shows that the projector $\frac{1}{2}(I+\mathrm{SWAP})$ onto the symmetric subspace of the two registers equals the sum, over all even-size subsets $E\subseteq[q]$, of $\bigotimes_{t\in E}\Pi^-_t\otimes\bigotimes_{t\notin E}(I-\Pi^-_t)$. The Bell-basis measurements measure the commuting projectors $\Pi^-_1,\ldots,\Pi^-_q$, so the probability of an even number of singlet outcomes is $\bra{\psi}\!\bra{\phi}\frac{1}{2}(I+\mathrm{SWAP})\ket{\psi}\!\ket{\phi}=\frac12+\frac12|\braket{\psi}{\phi}|^2$.
The destructive swap test uses $O(\log d)$ gates, all of which act on disjoint pairs of qubits, so it has constant depth for every local dimension $d$.
\begin{boxedalgorithm}[H]
    \caption{$k$-fold SWAP test.}
    \label{alg:k-fold-swap-test}
    
    \vspace{0.5em}
    \textbf{Input:}
    Two unknown pure states $\ket{\psi_i}$ and $\ket{\psi_j}$,
    $k$ copies of each state, and a threshold $\tau$.
    \\
    
    \textbf{Output:}
    \texttt{accept} or \texttt{reject}.
    \\
    
    \textbf{Procedure:}
    \begin{enumerate}
        
        \item Perform $k$ independent (destructive) SWAP tests using fresh copies of
        $\ket{\psi_i}$ and $\ket{\psi_j}$.
        
        \item Let $X_{ij}^{(s)}\in\{0,1\}$ indicate whether the $s$-th
        SWAP test produces the symmetric outcome, and define
        \begin{align}
            \widehat p_{ij}
            :=
            \frac{1}{k}\sum_{s=1}^k X_{ij}^{(s)}.
        \end{align}
        
        \item If $\widehat p_{ij}<\tau$, return \texttt{reject};
        otherwise, return \texttt{accept}.
    
    \end{enumerate}

\end{boxedalgorithm}

\begin{lemma}\label{lem:k-fold-swap}
    For $\epsilon\in(0,1/4)$ and $\gamma=1/2-\sqrt\epsilon$, the $k$-fold swap test with threshold $\tau=\frac{7-2\epsilon+\sqrt{\epsilon}}{8}$ decides whether $|\braket{\psi}{\phi}|^2\geq 1-\epsilon$ or $|\braket{\psi}{\phi}|^2\leq \frac{1+\sqrt\epsilon}{2}$ correctly with probability at least $1-e^{-k\gamma^2/8}$. In particular, it solves $\mathrm{QSI}_2^\epsilon$ with error at most $e^{-k\gamma^2/8}$.
\end{lemma}

\begin{proof}
Let $X_i$ be the indicator that the $i$-th SWAP test outputs the symmetric outcome, and let $X=\sum_{i=1}^k X_i$. Since a SWAP test accepts with probability $\frac{1}{2}+\frac{1}{2}|\langle\psi|\phi\rangle|^2$, let
\begin{align}
    p_{\mathrm{yes}}
    &=1-\frac{\epsilon}{2} \quad \text{and} \quad
    p_{\mathrm{no}}
    =\frac{3+\sqrt{\epsilon}}{4}.
\end{align}
For a yes instance, each swap test accepts with probability at least $p_\mathrm{yes}$, and for a no instance, each swap test accepts with probability at most $p_\mathrm{no}$.
The threshold $\tau$ is the midpoint of $p_{\mathrm{yes}}$ and $p_{\mathrm{no}}$; define
\begin{align}
    \Delta
    :=p_{\mathrm{yes}}-p_{\mathrm{no}}
    =\frac{1-2\epsilon-\sqrt{\epsilon}}{4}.
\end{align}
Thus, $p_{\mathrm{yes}}-\tau=\tau-p_{\mathrm{no}}=\Delta/2$.

For a yes instance, the test fails only if $X/k<\tau$. Since $\mathbf{E}[X]/k\geq p_{\mathrm{yes}}$, Hoeffding's inequality gives
\begin{align}
    \Pr[\mathrm{reject}\mid\mathrm{yes}]
    &\leq
    \Pr\left[
        \frac{1}{k}(X-\mathbf{E}[X])
        \leq-\frac{\Delta}{2}
    \right]
    \leq e^{-k\Delta^2/2}.
\end{align}
Similarly, for a no instance, $\mathbf{E}[X]/k\leq p_{\mathrm{no}}$, and applying the upper-tail bound gives
\begin{align}
    \Pr[\mathrm{accept}\mid\mathrm{no}]
    \leq e^{-k\Delta^2/2}.
\end{align}
Hence, the error probability in either case is at most $e^{-k\Delta^2/2}$.

Since $\gamma=1/2-\sqrt{\epsilon}$, we have $\sqrt{\epsilon}=1/2-\gamma$ and $\epsilon=(1/2-\gamma)^2$. Therefore,
\begin{align}
    \Delta
    &=\frac{1-2(1/2-\gamma)^2-(1/2-\gamma)}{4}=\frac{\gamma(3-2\gamma)}{4}.
\end{align}
Since $\epsilon\in(0,1/4)$, we have $\gamma\in(0,1/2)$, and thus $3-2\gamma\geq2$. Hence,
\(
    \Delta\geq\frac{\gamma}{2},
\)
and it follows that
\(
    e^{-k\Delta^2/2}
    \leq e^{-k\gamma^2/8}.
\)
Therefore, the $k$-fold swap test succeeds with probability at least $1-e^{-k\gamma^2/8}$.
\end{proof}

In particular, for $n=2$, the $k$-fold swap test with $k=\left\lceil\frac{8}{\gamma^2}\log(2/\delta)\right\rceil$ solves $\mathrm{QSI}_2^\epsilon$ with error at most $\delta$ using $k$ copies of each input state, which establishes \thm{informal-random-matching} in this case. For $n\ge 3$ we use the random-matching protocol of the next subsection.

\subsection{Extension by random matching}\label{sec:random-matching}

For $n\geq 3$, we sample a random matching and test each pair in the matching using the $k$-fold swap test (see the description in \alg{random-matching}).
\begin{boxedalgorithm}[H]
    \caption{Random-Matching Protocol for the relaxed quantum state
    identity problem for $n\ge 3$.}
    \label{alg:random-matching}
    
    \vspace{0.5em}
    \textbf{Input:}
    Unknown pure quantum states
    $\ket{\psi_1},\ldots,\ket{\psi_n}$, where $n\ge 3$ (if $n$ is odd, first append an additional register holding an extra copy of $\ket{\psi_1}$, so that $n$ becomes even);
    a tolerance parameter $\epsilon\in(0,\frac{1}{4})$;
    a target error probability $\delta\in(0,1)$.
    \\
    
    \textbf{Output:}
    \texttt{accept} or \texttt{reject} with success probability at least $1-\delta$.
    \\
    
    \textbf{Procedure:}
    \begin{enumerate}

        \item Set the decision threshold
            \(\tau:=\frac{7-2\epsilon+\sqrt{\epsilon}}{8}\).

        \item Set $\ell = \lceil\log_3(2/\delta)\rceil$, so that the random matchings fail to detect a witness edge with probability at most $\delta/2$.

        \item Set $\gamma=\frac12-\sqrt\epsilon$ and $k=\left\lceil\frac{8}{\gamma^2}\cdot\log(n\ell/\delta)\right\rceil$.
        \item Independently sample $\ell$ uniformly random perfect matchings $M_1,\ldots,M_\ell$ of $[n]$. 
        
        \begin{enumerate}
            \item For each round $r\in[\ell]$ and each edge $\{i,j\}\in M_r$, run \alg{k-fold-swap-test} on $\ket{\psi_i}$ and $\ket{\psi_j}$ with threshold $\tau$; each edge consumes $k$ copies of each of the two states.\label{itm:a}

            \item If any call in Step~\ref{itm:a} returns \texttt{reject}, return \texttt{reject}.

        \end{enumerate}
        \item Return \texttt{accept}.
     
    \end{enumerate}
    
\end{boxedalgorithm}

\begin{lemma}\label{lem:random-matching}
    Let $n\geq 4$ be even, and let $A,B$ be two non-empty sets that form a partition of $[n]$.
    Let $i\in A$ and $j\in B$ be two fixed elements.
    When $\ell$ independent uniformly random perfect matchings of $[n]$ are sampled, the probability that in none of them $i$ is matched to an element of $B$ or $j$ is matched to an element of $A$ is at most $(1/3)^\ell$.
\end{lemma}
\begin{proof}
    Let $|A|=a$ and $|B|=b$.
    In a uniformly random perfect matching, $i$ is matched inside $A$ with probability $\frac{a-1}{n-1}$. Conditioned on that event, $j$ is matched inside $B$ with probability $\frac{b-1}{n-3}$.
    Therefore, the probability that in a single random matching $i$ is not matched to $B$ and $j$ is not matched to $A$ is
    \begin{align}
        \mathrm{Pr}[\mathrm{unmatched}] = \frac{(a-1)(b-1)}{(n-1)(n-3)} \leq
        \frac{\frac{(a-1+b-1)^2}{4}}{(n-1)(n-3)}= \frac{(n-2)^2}{4(n-1)(n-3)} =\frac{1}{4}\cdot\frac{1}{1-\frac{1}{(n-2)^2}} \le \frac{1}{3},
    \end{align}
    where the first inequality is the AM--GM inequality and the last one uses $n\ge 4$.
    The lemma follows since the $\ell$ matchings are independent.
\end{proof}

If $n$ is odd, we first append one additional register containing an extra copy of $\ket{\psi_1}$, and run the protocol on the resulting $n+1$ registers; this doubles the number of copies of $\ket{\psi_1}$ that are consumed. Both promises are preserved: in a yes instance, the new register has fidelity $1$ with $\ket{\psi_1}$ and fidelity at least $1-\epsilon$ with every other state, and in a no instance the bad pair is still present. In the rest of this section we therefore assume that $n$ is even.

Let $(i,j)$ be a bad pair of a no instance, i.e., $F_{ij}\le\epsilon$. By \lem{geometric-witness}, every other register is far from $\ket{\psi_i}$ or from $\ket{\psi_j}$.
Thus we can partition the index set $[n]$ into two sets $S_i$ and $S_j$, where $S_j$ consists of the registers that are far from $\ket{\psi_i}$ (in particular $j\in S_j$) and $S_i$ consists of the registers that are far from $\ket{\psi_j}$ (in particular $i\in S_i$), with ties broken arbitrarily.
Then \lem{random-matching} implies that, when $\ell$ independent random matchings are sampled, with probability at least $1-3^{-\ell}$ at least one edge connects an index in $\{i,j\}$ to an index whose corresponding state has fidelity at most $\frac{1+\sqrt\epsilon}{2}$ with it.
On the other hand, for a yes instance, it is guaranteed that all edges in the matchings have fidelity at least $1-\epsilon$.

\paragraph{Soundness.}
When the algorithm takes as input a no instance, it is guaranteed that there is a pair of distinct indices $i,j$ such that $|\langle\psi_i|\psi_j\rangle|^2\leq \epsilon$.
Our algorithm samples $\ell=O\left(\log(2/\delta)\right)$ random matchings, so that with probability at least $1-\delta/2$, either $i$ or $j$ is paired with an index whose corresponding state has fidelity at most $\eta=\frac{1+\sqrt\epsilon}{2}$ with it.
By \lem{k-fold-swap}, for $\epsilon=(1/2-\gamma)^2$, running the $k$-fold swap test on each edge of the matchings with $k=O\left(\frac{1}{\gamma^2}\log(2/\delta)\right)$ copies ensures that such a pair, whose fidelity is at most $\eta$, is accepted with probability at most $\delta/2$.
Then by a union bound, the overall soundness error is at most $\delta$, and $k\ell = O\left(\frac{1}{\gamma^2}\log^2(2/\delta)\right)$ copies of each state suffice for soundness.

\begin{theorem}[Soundness]\label{thm:randommatching-soundness}
    For $\epsilon< 1/4$, \alg{random-matching} rejects any no instance with probability at least $1-\delta$. Moreover, for soundness alone, the choices $\ell=\lceil\log_3(2/\delta)\rceil$ and $k=\left\lceil\frac{8}{\gamma^2}\log(2/\delta)\right\rceil$ already suffice, i.e., $O\left(\frac{1}{(1-2\sqrt\epsilon)^2}\log^2(2/\delta)\right)$ copies of each input state (at most $2k\ell$ copies of $\ket{\psi_1}$ and $k\ell$ copies of every other state when $n$ is odd).
\end{theorem}
\begin{proof}
    Given a no instance, there exists a pair of indices $i<j$ such that $\ket{\psi_i}$ and $\ket{\psi_j}$ have fidelity at most $\epsilon$.
    Then by \lem{geometric-witness}, for every other index $s\in[n]\setminus\{i,j\}$, at least one of the pairs $(i,s)$ and $(j,s)$ has fidelity at most $\eta=\frac{1+\sqrt{\epsilon}}{2}$.
    Let $S_j$ be the set of indices $s\in[n]\setminus\{i\}$ such that the fidelity of $(i,s)$ is at most $\eta$, and let $S_i$ be the set of indices $s\in[n]\setminus\{j\}$ such that the fidelity of $(j,s)$ is at most $\eta$ and $s\notin S_j$. Then $j\in S_j$, $i\in S_i$, and $(S_i,S_j)$ is a partition of $[n]$.

    By \lem{random-matching}, when $\ell$ random matchings are sampled, with probability at least $1-(1/3)^\ell$, either $i$ is matched to an index in $S_j$ or $j$ is matched to an index in $S_i$.
    For $\ell\ge\log_{3}(2/\delta)$, this probability is at least $1-\delta/2$.
    Conditioned on that event, some sampled edge $(u,v) \in \left(\{i\}\times S_j\right) \cup \left(\{j\}\times S_i\right)$ satisfies $|\langle\psi_u|\psi_v\rangle|^2\leq \frac{1+\sqrt\epsilon}{2}$.
    By \lem{k-fold-swap}, the $k$-fold swap test with threshold $\tau=\frac{7-2\epsilon+\sqrt\epsilon}{8}$ accepts the pair $(u,v)$ with probability at most $e^{-k\gamma^2/8}$, where $\gamma=1/2-\sqrt\epsilon$.
    For $k\ge\frac{8}{\gamma^2}\cdot\log(2/\delta)$, this probability is at most $\delta/2$.
    Finally, by a union bound, the soundness error is at most $\delta$. Since the parameters chosen in \alg{random-matching} satisfy both conditions, the theorem follows.
\end{proof}

\paragraph{Completeness.}
For the completeness analysis, it suffices to show that, using sufficiently many copies, the algorithm accepts every edge in the random matchings.
\begin{theorem}[Completeness]\label{thm:randommatching-completeness}
    For $\epsilon< 1/4$, \alg{random-matching} accepts any yes instance with probability at least $1-\delta$, using $k\ell=O\left(\frac{1}{(1-2\sqrt\epsilon)^2}\log(2/\delta)\cdot\log(n/\delta)\right)$ copies of each input state (at most $2k\ell$ copies of $\ket{\psi_1}$ when $n$ is odd).
\end{theorem}

\begin{proof}
    When the algorithm takes as input a yes instance, it is guaranteed that every pair of states has fidelity at least $1-\epsilon$.
Thus, in the $\ell$ sampled random matchings, each edge is accepted by the $k$-fold swap test with probability at least $1-e^{-k\gamma^2/8}$, where $\gamma=\frac12-\sqrt{\epsilon}$ as in the soundness analysis.
Since there are $\ell n/2$ edges in total, by a union bound the error probability is at most $\frac{\ell n}{2}\cdot e^{-k\gamma^2/8}$.
Therefore, with $\ell=\lceil\log_3(2/\delta)\rceil$ as required by the soundness analysis and $k\ge\frac{8}{\gamma^2}\log(n\ell/\delta)$ as chosen in \alg{random-matching}, the completeness error is at most $\delta/2$. Since $\log(n\ell/\delta)=\log(n/\delta)+\log\ell=O(\log(n/\delta))$, we have $k=O\left(\frac{1}{\gamma^2}\log(n/\delta)\right)$, which gives the claimed number of copies.
\end{proof}

\paragraph{Complexity of \alg{random-matching}.}
Regarding sample complexity, \alg{random-matching} solves $\mathrm{QSI}_{n}^\epsilon$ with probability $1-\delta$ using $O\left(\frac{1}{\gamma^2}\log(n/\delta)\log(2/\delta)\right)$ copies of each input state. 
As we will see in \sec{lower-bound}, this sample complexity is optimal when $\delta$ and $\epsilon$ are constants.
However, for $\delta=o(1)$, compared with our lower bound in \sec{lower-bound}, the sample complexity incurs an additional multiplicative $\log(2/\delta)$ factor.

Regarding gate complexity and circuit depth, the algorithm runs $\frac{n}{2}\cdot k\ell=O(\frac{n}{\gamma^2}\log(n/\delta)\log(2/\delta))$ swap tests in parallel, each on two registers of $\lceil\log_2 d\rceil$ qubits. Each destructive swap test (\sec{swap-test}) uses $O(\log d)$ gates acting on disjoint pairs of qubits, followed by measurements and classical post-processing. 
Therefore, the gate complexity is $O(\frac{n\log d}{\gamma^2}\log(n/\delta)\log(2/\delta))$ and the depth is $O(1)$, independently of $n$, $d$, $\epsilon$, and $\delta$.

\section{A sample-optimal algorithm}\label{sec:optimal}
In this section we present an algorithm whose sample complexity is optimal when the tolerance parameter $\epsilon$ is a constant. Given $n$ quantum states $\ket{\psi_1}, \ldots,\ket{\psi_n}$, let $\rho = \ket{\psi_1\cdots \psi_n}\bra{\psi_1\cdots \psi_n}$. The algorithm applies the Schur transform to $\rho$ and measures the symmetric-group register, as described in \alg{jm-protocol}.

\begin{boxedalgorithm}[H]
\caption{Jucys--Murphy protocol for the relaxed quantum state
identity problem.}
\label{alg:jm-protocol}

\vspace{0.5em}
\textbf{Input:}
Unknown pure quantum states
$\ket{\psi_1},\ldots,\ket{\psi_n}$;
a tolerance parameter $\epsilon\in(0,\frac14)$;
a target error probability $\delta\in(0,1)$.
\\

\textbf{Output:}
\texttt{accept} or \texttt{reject} with success probability at least
$1-\delta$.
\\

\textbf{Procedure:}
\begin{enumerate}

    \item Set the decision threshold $\tau :=\frac{3n-4}{4}-\frac{n-2}{2}\epsilon+\frac{n-2}{4}\sqrt{\epsilon}.$

    \item Set $\gamma = \frac{1}{2} - \sqrt{\epsilon}$ and $L:=\left\lceil\frac{8}{\gamma^2}\log\frac{2n}{\delta}\right\rceil.$ Each Schur transform below is implemented by a circuit whose channel is within diamond-norm distance $\xi:=\frac{\delta}{4L}$ of the exact Schur transform.

    \item For each round $r\in[L]$, use two fresh copies of every input
    state and perform the following:

    \begin{enumerate}

        \item On the forward-ordered block
        $\ket{\psi_1}\otimes\cdots\otimes\ket{\psi_n}$,
        apply the Schur transform and measure the symmetric-group register in Young's orthogonal basis, obtaining
        $(\lambda_r,S_r)$.
        For every $i\in[n]$, set
        $C_{i,r}^{\rightarrow}:=\operatorname{ct}_{S_r}(i)$, where $\operatorname{ct}_{S_r}(i)$ is the content of the box containing $i$ (\dfn{content}).

        \item On the reverse-ordered block
        $\ket{\psi_n}\otimes\cdots\otimes\ket{\psi_1}$,
        apply the Schur transform and measure the symmetric-group register in Young's orthogonal basis, obtaining
        $(\lambda_r',S'_r)$.
        For every $i\in[n]$, set
        $C_{i,r}^{\leftarrow}:=\operatorname{ct}_{S'_r}(n+1-i)$.

        \item For every $i\in[n]$, set
        $Z_{i,r}:=C_{i,r}^{\rightarrow}+C_{i,r}^{\leftarrow}$.

    \end{enumerate}

    \item For every $i\in[n]$, compute
    \begin{align}
    \bar{Z}_i
    :=
    \frac{1}{L}
    \sum_{r=1}^{L} Z_{i,r}.
    \end{align}

    \item If $\min_{i\in[n]}\bar{Z}_i<\tau$, return \texttt{reject}. Otherwise, return \texttt{accept}.

\end{enumerate}

\end{boxedalgorithm}

After applying the Schur transform, we measure the Young diagram
label $\lambda$ and the standard Young tableau label $S$, where $\ket{S}$ belongs to Young's orthogonal basis of the symmetric-group register, while leaving the
unitary register unmeasured.  The corresponding projectors on the
original tensor-product space are
\begin{align}
M_{\lambda,S}=U_{\mathrm{Schur}}^\dagger\left(\ket{\lambda}\!\bra{\lambda}\otimes\ket{S}\!\bra{S}\otimes I_{\operatorname{dim}(V_\lambda^d)}\right)U_{\mathrm{Schur}}.
\end{align}
Equivalently, this measurement is obtained by measuring the full Schur
basis $\ket{\lambda,S,T}$ and marginalizing over the unitary-register
label $T$.
The operators $\{M_{\lambda,S}\}_{\lambda,S}$ form a projective
measurement, since the corresponding Schur-basis projectors are
mutually orthogonal and sum to the identity, and these properties are
preserved under unitary conjugation.

We next exploit the structure of the \textit{Jucys--Murphy elements}
in \dfn{jucysmurphy}. By \lem{jucys-murphy-eigenvalue}, under the
Schur--Weyl decomposition,
\begin{align}
    U_{\mathrm{Schur}}P(J_i)U_{\mathrm{Schur}}^\dagger=\sum_{\substack{\lambda\vdash n\\ \ell(\lambda)\le d}} \ket{\lambda}\!\bra{\lambda}\otimes\left(\sum_{S\in\operatorname{SYT}(\lambda)}\operatorname{ct}_S(i)\ket{S}\!\bra{S}\right)\otimes
    I_{\operatorname{dim}(V_\lambda^d)}.
\end{align}
Hence, the eigenvalue of $P(J_i)$ is determined by the tableau label $S$
and is independent of the unitary-register label $T$.

\begin{corollary}[Joint spectral measurement of the Jucys--Murphy operators]
\label{cor:jm-joint-measurement}
For every $i\in[n]$,
\begin{align}\label{Jucys-murph-opeartor}
P(J_i)
=
\sum_{\substack{\lambda\vdash n\\ \ell(\lambda)\le d}}
\sum_{S\in\operatorname{SYT}(\lambda)}
\operatorname{ct}_S(i)\,
M_{\lambda,S}.
\end{align}
Therefore, the projective measurement
$\{M_{\lambda,S}\}_{\lambda,S}$ jointly resolves the eigenvalues of
the commuting Jucys--Murphy operators. In particular, an outcome
$(\lambda,S)$ determines the eigenvalue $\operatorname{ct}_S(i)$ of
each $J_i$.
\end{corollary}

The above corollary shows that a single outcome $(\lambda,S)$ of the
joint spectral measurement determines the eigenvalues of all the
Jucys--Murphy operators through the tableau contents
$\operatorname{ct}_S(i)$. We now relate these measured contents to the
pairwise fidelities of the input states. Since each $J_i$ is a sum of
SWAP operators, its expectation value can be expressed as a sum of
pairwise fidelities. By applying the measurement to both the forward
and reverse orderings of the input states, we obtain a bounded random
variable whose expectation is the sum of all fidelities involving a
given state.

\begin{lemma}[Fidelity-sum estimator]\label{lem:contentexpectation}
Given $n$ pure quantum states
$\ket{\psi_1},\ldots,\ket{\psi_n}$, for every $i\in[n]$ there exists a
random variable $Z_i$ obtained from two joint spectral measurements of
the Jucys--Murphy operators such that   
\begin{align}
\mathbb{E}[Z_i]=\sum_{j\neq i} F_{ij},
\end{align}
where $F_{ij}=|\braket{\psi_i}{\psi_j}|^2$. Moreover, $-(n-1)\le Z_i\le n-1.$
\end{lemma}

\begin{proof}
First consider the input states in the forward order $\ket{\psi_1}\otimes\cdots\otimes\ket{\psi_n}.$
Let $(\lambda,S)$ be the outcome of the joint spectral measurement
$\{M_{\lambda,S}\}_{\lambda,S}$, and define $C_i^\rightarrow:=\operatorname{ct}_S(i).$ The probability of obtaining $(\lambda,S)$ is $p(\lambda,S)=\tr[\rho M_{\lambda,S}]$. By \cor{jm-joint-measurement},
\begin{align}
    \mathbb{E}[C_i^\rightarrow]
    &=\sum_{\lambda,S}\operatorname{ct}_S(i)\tr[\rho M_{\lambda,S}]\\
    &= \tr[\rho\sum_{\lambda,S}\operatorname{ct}_S(i)M_{\lambda,S}] \\
    &=\tr(J_i\rho). 
\end{align}
Using \dfn{jucysmurphy} and $\tr[(\ket{\psi_i}\!\bra{\psi_i}\otimes\ket{\psi_j}\!\bra{\psi_j})\,\mathrm{SWAP}]=|\braket{\psi_i}{\psi_j}|^2$,
\begin{align}
    \mathbb{E}[C_i^\rightarrow]
    &=
    \sum_{j<i}\tr[\rho P_{ij}]\\
    &=
    \sum_{j<i}F_{ij}.
\end{align}

Next consider a second block with the input states in reverse order, $\ket{\psi_n}\otimes\cdots\otimes\ket{\psi_1}.$
Let $S'$ denote the tableau obtained from the corresponding joint
spectral measurement, and define $C_i^\leftarrow:=\operatorname{ct}_{S'}(n+1-i).$ Applying the same argument to the reversed ordering gives
\begin{align}
    \mathbb{E}[C_i^\leftarrow]
    =
    \sum_{j>i}F_{ij}.
\end{align}

Finally, define $Z_i:=C_i^\rightarrow+C_i^\leftarrow.$ Therefore,
\begin{align}
\mathbb{E}[Z_i]
&=
\mathbb{E}[C_i^\rightarrow]
+
\mathbb{E}[C_i^\leftarrow]\\
&=
\sum_{j<i}F_{ij}
+
\sum_{j>i}F_{ij}\\
&=
\sum_{j\neq i}F_{ij}.
\end{align}

Since the box containing $i$ in a standard Young tableau lies in the first $i$ rows and the first $i$ columns, we have $-(i-1)\le C_i^\rightarrow\le i-1$; similarly, $-(n-i)\le C_i^\leftarrow\le n-i$. Hence
\begin{align}
-(n-1)\le Z_i\le n-1.
\end{align}

\end{proof}
\lem{contentexpectation} shows that the Jucys--Murphy measurement
provides a bounded estimator for the total fidelity involving each
input state. In particular, in each round $r\in[L]$, \alg{jm-protocol} obtains a random variable $Z_{i,r}$ for every
$i\in[n]$. Since every round is performed on fresh copies of the input states,
for every fixed $i\in[n]$, the random variables
$Z_{i,1},\ldots,Z_{i,L}$ are independent and identically distributed.
Therefore, defining
\begin{align}
\bar{Z_i}:=\frac{1}{L}\sum_{r=1}^{L}Z_{i,r},
\end{align}
we can apply Hoeffding's inequality to show that $\bar Z_i$
concentrates around $\mathbb{E}[Z_i]$.
We will use these concentration bounds to separate the yes and no
cases.

In the yes case, we have the promise that every pair $(i,j)$ has fidelity at least $1-\epsilon$. Therefore, for every $i\in[n]$,
\begin{align}\label{eq:jm-yes}
    \mathbb{E}[Z_i] = \sum_{j\neq i}F_{ij} \ge (n-1)(1-\epsilon).
\end{align}
In the no case, let $(a,b)$ be a bad pair, so that $F_{ab}\le\epsilon$. For every $j\neq a,b$, we have
\begin{align}\label{eq:spectral-pair-bound}
    F_{aj}+F_{bj}=\bra{\psi_j}\big(\ket{\psi_a}\!\bra{\psi_a}+\ket{\psi_b}\!\bra{\psi_b}\big)\ket{\psi_j}\le \big\|\ket{\psi_a}\!\bra{\psi_a}+\ket{\psi_b}\!\bra{\psi_b}\big\|_{\infty} = 1+|\braket{\psi_a}{\psi_b}|\le 1+\sqrt{\epsilon},
\end{align}
since the nonzero eigenvalues of $\ket{\psi_a}\!\bra{\psi_a}+\ket{\psi_b}\!\bra{\psi_b}$ are $1\pm|\braket{\psi_a}{\psi_b}|$. Summing over $j$, the two expectation values satisfy
\begin{align}
    \mathbb{E}[Z_a + Z_b] = 2F_{ab}+\sum_{j\neq a,b}(F_{aj}+F_{bj}) \le 2\epsilon + (n-2)(1+\sqrt{\epsilon}).
\end{align}
Hence $\min\{\mathbb{E}[Z_a],\mathbb{E}[Z_b]\}\le\frac12\mathbb{E}[Z_a+Z_b]$, so there exists an index $i^*\in\{a,b\}$ such that
\begin{align}\label{eq:jm-no}
    \mathbb{E}[Z_{i^*}]\le \epsilon+\frac{n-2}{2}(1+\sqrt{\epsilon}).
\end{align}
Note that the algorithm does not need to know $a$, $b$, or $i^*$: it simply checks whether $\min_{i\in[n]}\bar Z_i$ falls below the threshold.
From \eq{jm-yes} and \eq{jm-no}, with $\gamma =\frac{1}{2}-\sqrt{\epsilon}$, the gap between the two cases is
\begin{align}
    \Delta &:= (n-1)(1-\epsilon)-\left(\epsilon+\frac{n-2}{2}(1+\sqrt{\epsilon})\right) \\
    &= \frac{n}{2}\left(1-2\epsilon-(1-\frac{2}{n})\sqrt{\epsilon}\right)\\
    &\ge \frac{n}{2}(1-2\epsilon-\sqrt{\epsilon})\\
    &= \frac{n}{2}\gamma(3-2\gamma)\ge n\gamma.
\end{align}
We now follow the same analysis as in \lem{k-fold-swap}. Let $\mu_{\mathrm{yes}}:= (n-1)(1-\epsilon)$ and $\mu_{\mathrm{no}}:=\epsilon+\frac{n-2}{2}(1+\sqrt{\epsilon})$. The threshold $\tau$ in \alg{jm-protocol} is the midpoint of $\mu_{\mathrm{yes}}$ and $\mu_{\mathrm{no}}$, so that
\begin{align}
    \frac{\Delta}{2} = \mu_{\mathrm{yes}}- \tau = \tau -\mu_{\mathrm{no}}.
\end{align}
To distinguish the two cases, we apply Hoeffding's inequality to each case.

\begin{theorem}[Soundness] \label{thm:jucysmurphy-soundness}
    For $\epsilon < \frac{1}{4}$, \alg{jm-protocol} with exact Schur transforms rejects any no instance with probability at least $1-\delta/2$. Moreover, for soundness alone, $L=\left\lceil\frac{8}{\gamma^2}\log(2/\delta)\right\rceil$ rounds already suffice, i.e., $O\left(\frac{1}{(1-2\sqrt{\epsilon})^2}\log(2/\delta)\right)$ copies of each input state.
\end{theorem}
\begin{proof}
    Given a no instance, the protocol accepts only if $\min_{i\in[n]}\bar{Z}_i\ge\tau$, which implies $\bar{Z}_{i^*} \ge \tau$ and hence $\bar{Z}_{i^*}-\mathbb{E}[Z_{i^*}]\ge\Delta/2$ by \eq{jm-no}.
    Since $Z_{i^*,1},\ldots,Z_{i^*,L}$ are i.i.d.\ and take values in $[-(n-1),n-1]$ by \lem{contentexpectation}, Hoeffding's inequality gives
    \begin{align}
    \mathrm{Pr}\left[\frac{1}{L}\sum^L_{r=1}Z_{i^*,r} - \mathbb{E}[Z_{i^*}]\ge \frac{\Delta}{2}\right] \le \exp\left(-\frac{2L(\Delta/2)^2}{4(n-1)^2}\right) \le \exp\left(-\frac{L\cdot n^2\gamma^2}{8(n-1)^2}\right)\le e^{-L\gamma^2/8}.
\end{align}
To bound the error probability by $\delta/2$, it thus suffices to take $L \ge \frac{8}{\gamma^2}\log(2/\delta)$ rounds, i.e., $2L = O\left(\frac{1}{\gamma^2}\log(2/\delta)\right)$ copies of each state. The choice of $L$ in \alg{jm-protocol} satisfies this condition.

\end{proof}

\begin{theorem}[Completeness]\label{thm:jucysmurphy-completeness}
     For $\epsilon < \frac{1}{4}$, \alg{jm-protocol} with exact Schur transforms accepts any yes instance with probability at least $1-\delta/2$, using $2L=O\left(\frac{1}{(1-2\sqrt{\epsilon})^2}\log(n/\delta)\right)$ copies of each input state.
\end{theorem}
\begin{proof}
In the yes case, the protocol rejects only if $\bar{Z}_i < \tau$ for some $i\in[n]$, which implies $\bar Z_i-\mathbb{E}[Z_i]\le-\Delta/2$ by \eq{jm-yes}. For each fixed $i$, Hoeffding's inequality gives
\begin{align}
    \mathrm{Pr}\left[\frac{1}{L}\sum^L_{r=1}Z_{i,r} - \mathbb{E}[Z_{i}]\le -\frac{\Delta}{2}\right] \le \exp\left(-\frac{2L(\Delta/2)^2}{4(n-1)^2}\right) \le \exp\left(-\frac{L\cdot n^2\gamma^2}{8(n-1)^2}\right)\le e^{-L\gamma^2/8}.
\end{align}
By a union bound over all $i\in[n]$, the error probability is
\begin{align}
    \mathrm{Pr}[\mathtt{reject}\mid\mathrm{yes}] \le n\cdot e^{-L\gamma^2/8}.
\end{align}
With $L = \left\lceil\frac{8}{\gamma^2}\log(2n/\delta)\right\rceil$ as in \alg{jm-protocol}, the error probability is at most $\delta/2$, and the protocol uses $2L=O\left(\frac{1}{\gamma^2}\log(n/\delta)\right)$ copies of each state.
\end{proof}

\paragraph{Approximate Schur transforms.}
The analysis above assumes that the Schur transforms are implemented exactly. \alg{jm-protocol} applies $2L$ Schur transforms in total, each implemented by a circuit whose channel is within diamond-norm distance $\xi=\delta/(4L)$ of the exact transform. Replacing the exact transforms by their approximations one at a time, the triangle inequality shows that the probability of any outcome of the protocol changes by at most $2L\xi=\delta/2$. Together with \thm{jucysmurphy-soundness} and \thm{jucysmurphy-completeness}, the implemented protocol therefore solves $\mathrm{QSI}_n^\epsilon$ with error at most $\delta/2+\delta/2=\delta$, using $2L=O\left(\frac{1}{\gamma^2}\log(n/\delta)\right)$ copies of each input state. Since the cost of the Schur transform depends only polylogarithmically on the inverse precision $1/\xi=O\left(\frac{1}{\delta\gamma^2}\log(n/\delta)\right)$, the required precision affects only the polylogarithmic factors hidden in the gate complexity below.

\paragraph{Complexity of \alg{jm-protocol}.}
\alg{jm-protocol} uses $O\left(\frac{1}{(1-2\sqrt\epsilon)^2}\log(n/\delta)\right)$ copies of each input state to solve $\mathrm{QSI}_{n}^\epsilon$ with probability $1-\delta$.
Compared to \alg{random-matching}, \alg{jm-protocol} improves the sample complexity by a multiplicative factor of $\log(2/\delta)$.
As we will see in \sec{lower-bound}, this algorithm is optimal in terms of sample complexity for constant $\epsilon$.

Regarding the gate complexity and depth,
\alg{jm-protocol} runs the Schur transform $2L=O(\frac{1}{\gamma^2}\log(n/\delta))$ times in parallel to learn the expectations of the Jucys--Murphy elements. By the analysis of Burchardt et al.~\cite{burchardtHighdimensionalQuantumSchur2025}, the circuit of Bacon, Chuang, and Harrow \cite{baconQuantumSchurClebschGordan2007} has gate complexity $\tilde O\left(\min\{n^5, nd^4\}\right)$, and the corrected version of Krovi's algorithm \cite{burchardtHighdimensionalQuantumSchur2025} has gate complexity $\tilde O(n^4)$, where $\tilde O$ hides polylogarithmic factors in $n$, $d$, and the inverse precision. Hence the Schur transform can be implemented with $\tilde O\left(\min\{n^4, nd^4\}\right)$ gates, and therefore also in depth $\tilde O\left(\min\{n^4, nd^4\}\right)$.
Therefore, the gate complexity of \alg{jm-protocol} is $\tilde O\left(\frac{1}{\gamma^2}\log(n/\delta)\cdot \min\{n^4, nd^4\}\right)$, and the depth is $\tilde O(\min\{n^4, nd^4\})$.

\section{Sample complexity lower bound}\label{sec:lower-bound}

In this section, we prove a sample complexity lower bound for the quantum state identity problem. We show that, for $0<\delta\le 1/2-c$ with a fixed constant $c>0$, any quantum algorithm that solves the problem with probability at least $1-\delta$ must use $\Omega\left(\frac{\log(n/\delta)}{\log(1/\epsilon)}\right)$ copies of each input state.
For constant $\epsilon$, our lower bound is tight, since in the previous section we have shown that there is a quantum algorithm that solves the problem using $O(\log(n/\delta))$ copies of each input state.

To prove the lower bound, we consider the following problem.
\begin{problem}\label{prob:hard}
    Fix an angle $\alpha\in[0,\pi/2]$. For $\varphi\in[0,2\pi)$, let $\ket{\psi_{\varphi}}:=\cos\alpha\ket{0} + e^{\i\varphi}\sin\alpha\ket{1}$ and $\ket{\phi_\varphi}:=\sin\alpha\ket{0} + e^{\i\varphi}\cos\alpha\ket{1}$.
    The problem is to distinguish between the following two cases:
    \begin{enumerate}
        \item Yes case: for some $\theta=(\theta_1,\ldots,\theta_n)\in[0,2\pi)^n$, the given quantum state is of the form
        \begin{align}
            \ket{\Psi_0(\theta)} := \ket{\psi_{\theta_1}}\otimes \cdots \otimes \ket{\psi_{\theta_n}}.
        \end{align}
        \item No case: for some $\theta\in[0,2\pi)^n$ and some $i\in[n]$, the given quantum state is of the form
        \begin{align}
            \ket{\Psi_{1i}(\theta)} := \ket{\psi_{\theta_1}}\otimes \cdots \otimes \ket{\psi_{\theta_{i-1}}}\otimes \ket{\phi_{\theta_i}}\otimes \ket{\psi_{\theta_{i+1}}}\otimes \cdots \otimes \ket{\psi_{\theta_n}}.
        \end{align}
    \end{enumerate}
\end{problem}

The following theorem shows that for sufficiently small $\alpha$, \prob{hard} is a special case of the quantum state identity problem with overlap parameter $\epsilon$.

\begin{theorem}\label{thm:hard-qsi}
    For $\alpha\leq\frac{1}{2}\arccos\sqrt{1-\epsilon}$, \prob{hard} is a special case of $\mathrm{QSI}_n^\epsilon$.
\end{theorem}
\begin{proof}
In the yes case, for every $\theta_i,\theta_j\in[0,2\pi)$,
\begin{align}
    |\bra{\psi_{\theta_i}}\psi_{\theta_j}\rangle|^2 = |\cos^2\alpha+\sin^2\alpha\, e^{\i(\theta_j-\theta_i)}|^2 \geq (\cos^2\alpha-\sin^2\alpha)^2=\cos^22\alpha \geq 1-\epsilon.
\end{align}
In the no case, for $j\neq i$,
\begin{align}
    |\bra{\phi_{\theta_i}}\psi_{\theta_j}\rangle|^2  = \sin^2\alpha\cos^2\alpha\, |1+e^{\i(\theta_j-\theta_i)}|^2 \leq 4\sin^2\alpha\cos^2\alpha = \sin^2 2\alpha \leq \epsilon.
\end{align}
Therefore, the promises of the quantum state identity problem hold on both sides, and thus \prob{hard} is a special case of $\mathrm{QSI}_n^\epsilon$.
\end{proof}

We first prove a lower bound that depends on the success probability but not on $n$.
The $n$ dependence will be established later by bounding the trace distance between the average states in the yes and no cases.
For a fixed coordinate $i$, any distinguisher that solves \prob{hard} must in particular be able to tell whether the $i$th coordinate of the quantum state is $\ket{\psi_{\theta_i}}$ or $\ket{\phi_{\theta_i}}$.
For $\alpha=\frac12\arccos\sqrt{1-\epsilon}$, the fidelity between $\ket{\psi_{\theta_i}}$ and $\ket{\phi_{\theta_i}}$ is $\sin^2 2\alpha=\epsilon$ exactly, so the fidelity between $m$ copies of these two states is $\epsilon^m$, and we can establish a lower bound on the number of copies $m$ required to distinguish them with probability at least $1-\delta$.

\begin{theorem}\label{thm:delta-bound}
    For $0<\delta<1/2$, any quantum algorithm that solves \prob{hard} with $\alpha=\frac12\arccos\sqrt{1-\epsilon}$ (and hence any quantum algorithm that solves $\mathrm{QSI}_n^\epsilon$) with probability at least $1-\delta$ must use at least $\frac{\log(1/(4\delta(1-\delta)))}{\log(1/\epsilon)}$ copies of each input state. In particular, it must use at least $\frac{\log(1/(4\delta))}{\log(1/\epsilon)}$ copies.
\end{theorem}
\begin{proof}
    Fix any $i\in[n]$. We can reduce the problem of distinguishing $\ket{\psi_{0}}$ from $\ket{\phi_{0}}$ to \prob{hard} by applying the phase gate $\mathrm{diag}(1,e^{\i\theta_i})$ for a uniformly random $\theta_i$ to each copy of the given state, embedding the resulting copies into the $i$-th coordinate, and preparing the remaining $n-1$ coordinates as copies of $\ket{\psi_{\theta_j}}$ for uniformly random $\theta_j$.
    Then any algorithm that solves \prob{hard} with probability at least $1-\delta$ can be used to distinguish $\ket{\psi_{0}}$ from $\ket{\phi_{0}}$ with the same probability.

    Now, given $m$ copies of either $\ket{\psi_{0}}$ or $\ket{\phi_{0}}$, the optimal distinguishing measurement is the Helstrom measurement, whose error probability is
    \begin{align}
        P_e = \frac{1}{2}\left(1-\sqrt{1-|\langle\psi_0|\phi_0\rangle|^{2m}}\right) = \frac{1}{2}\left(1-\sqrt{1-(\sin^22\alpha)^m}\right)=\frac{1}{2}\left(1-\sqrt{1-\epsilon^m}\right).
    \end{align}
    Since $\delta<1/2$, the condition $P_e\le\delta$ is equivalent to $\sqrt{1-\epsilon^m}\ge 1-2\delta$, i.e., to $\epsilon^m\le 4\delta(1-\delta)$, which gives $m\geq \frac{\log(1/(4\delta(1-\delta)))}{\log(1/\epsilon)}$. The second bound follows from $4\delta(1-\delta)\le 4\delta$.
\end{proof}

To obtain the dependence on $n$, observe that any quantum algorithm $\mathcal A$ that solves \prob{hard} with failure probability at most $\delta$ must accept $\ket{\Psi_0(\theta)}$ with probability at least $1-\delta$ averaged over uniformly random $\theta\in[0,2\pi)^n$, and must reject $\ket{\Psi_{1i}(\theta)}$ with probability at least $1-\delta$ averaged over uniformly random $i$ and $\theta$.
Thus we can establish a lower bound by bounding the trace distance between the average states in the yes and no cases.
To this end, we first prove the following lemma.

\begin{lemma}\label{lem:binomial}
    Let $\ket{\psi_\varphi}=\sqrt{1-p}\ket{0}+e^{\i\varphi}\sqrt{p}\ket{1}$ for $p\in[0,1]$. Then
    \begin{align}
        \Exp_{\varphi\sim\mathrm{Uniform}(0,2\pi)}[\proj{\psi_\varphi}^{\otimes m}] = \sum_{k=0}^m \binom{m}{k}p^{k}(1-p)^{m-k} \proj{\wt_k},
    \end{align}
    where $\ket{\wt_k}:=\binom{m}{k}^{-1/2}\sum_{x:|x|=k}\ket{x}$ is the uniform superposition of all $m$-bit strings of Hamming weight $k$.
\end{lemma}
\begin{proof}
    By direct calculation, we have
    \begin{align}
        \ket{\psi_\varphi}^{\otimes m}
        &= \sum_{x\in\bit^m} (1-p)^{\frac{m-|x|}{2}} p^{\frac{|x|}{2}} e^{\i\varphi |x|} \ket{x} \\
        &= \sum_{k=0}^m (1-p)^{\frac{m-k}{2}} p^{\frac{k}{2}} e^{\i\varphi k} \binom{m}{k}^{1/2}\ket{\wt_k}.
    \end{align}
    Taking the expectation over $\varphi\sim\mathrm{Uniform}(0,2\pi)$ removes the off-diagonal terms, leaving us with
    \begin{align}
        \Exp_{\varphi\sim\mathrm{Uniform}(0,2\pi)}[\proj{\psi_\varphi}^{\otimes m}]
        &= \sum_{k,\ell=0}^m \binom{m}{k}^{1/2}\binom{m}{\ell}^{1/2}p^{\frac{k+\ell}{2}}(1-p)^{m-\frac{k+\ell}{2}} \cdot \Exp_{\varphi}[e^{\i\varphi(k-\ell)}]\ket{\wt_k}\!\bra{\wt_{\ell}} \\
        &= \sum_{k=0}^m \binom{m}{k}p^{k}(1-p)^{m-k} \cdot \ket{\wt_k}\!\bra{\wt_{k}}.
    \end{align}
\end{proof}

\lem{binomial} implies that in the basis $\{\ket{\wt_k}: k=0,\ldots,m\}$, the state $\Exp_{\varphi\sim\mathrm{Uniform}(0,2\pi)}[\proj{\psi_\varphi}^{\otimes m}]$ is diagonal, and may be viewed as a binomial distribution with $m$ trials and success probability $p$.

By \lem{binomial}, we may view the task of distinguishing the average states as the following classical hypothesis testing problem.
\begin{problem}\label{prob:hypothesis}
Let $X_1,\ldots,X_n,Y$ be mutually independent with $X_i\sim\Binomial(m,p)$ and $Y\sim\Binomial(m,1-p)$, and let $I$ be uniformly distributed over $[n]$, independently of $(X_1,\ldots,X_n,Y)$.
\begin{itemize}
    \item Yes case: given $(X_1,\ldots,X_n)$.
    \item No case: same as the yes case, except that $X_I$ is replaced with $Y$.
\end{itemize}
\end{problem}
We now show that the total variation distance between the yes and no cases is small when $m$ is small.
\begin{theorem}\label{thm:tv-bound}
    Let $p\in(0,1)$, and let $X_1,\ldots,X_n,Y$ be mutually independent random variables with $X_i\sim\Binomial(m,p)$ for every $i\in[n]$ and $Y\sim\Binomial(m,1-p)$. Let $I$ be uniformly distributed over $[n]$ and independent of $(X_1,\ldots,X_n,Y)$, and let $Z_1,\ldots,Z_n$ be the same as $X_1,\ldots,X_n$ except that $Z_I$ is replaced with $Y$. Then the total variation distance between $(X_1,\ldots,X_n)$ and $(Z_1,\ldots,Z_n)$ is at most $\left(\frac{1}{4n}\left[\left(\frac{3p^2-3p+1}{p(1-p)}\right)^m - 1 \right]\right)^{1/2}$.
\end{theorem}
\begin{proof}
    Let $P$ be the probability mass function of $(X_1,\ldots,X_n)$, let $Q_i$ be the probability mass function of $(Z_1,\ldots,Z_n)$ conditioned on the event that the chosen coordinate is $i$, and let $Q$ be the probability mass function of $(Z_1,\ldots,Z_n)$. 
    Then by definition, $Q=\frac{1}{n}\sum_{i=1}^n Q_i$, and we have 
    \begin{align}
        \|P-Q\|_{TV}
        &= \frac{1}{2} \sum_{k_1,\ldots,k_n=0}^m \left|P(k_1,\ldots,k_n)-\frac{1}{n}\sum_{i=1}^n Q_i(k_1,\ldots,k_n)\right| \\
        &= \frac{1}{2}\Exp_{X_1,\ldots,X_n\sim\Binomial(m,p)}\left[\left|\frac{1}{n}\sum_{i=1}^n\frac{Q_i(X_1,\ldots, X_n)}{P(X_1,\ldots, X_n)} - 1\right|\right] \\
        &\leq \frac{1}{2}\Exp_{X_1,\ldots,X_n\sim\Binomial(m,p)}\left[\left(\frac{1}{n}\sum_{i=1}^n\frac{Q_i(X_1,\ldots, X_n)}{P(X_1,\ldots, X_n)} - 1\right)^2\right]^{1/2}.
    \end{align}
    where the last step is Jensen's inequality. For ease of notation, we define the likelihood ratios
    \begin{align}
        L_i(X_1,\ldots,X_n) := \frac{Q_i(X_1,\ldots, X_n)}{P(X_1,\ldots, X_n)}=\frac{(1-p)^{X_i}p^{m-X_i}}{p^{X_i}(1-p)^{m-X_i}}=\left(\frac{p}{1-p}\right)^{m-2X_i},
    \end{align}
    and write $L_i=L_i(X_1,\ldots,X_n)$ in the following analysis.
    The first two moments of $(L_1,\ldots,L_n)$ can be calculated explicitly as follows:
    \begin{align}
        \Exp_{X_1,\ldots,X_n\sim\Binomial(m,p)}[L_i] 
        = \sum_{k_1,\ldots,k_n=0}^m Q_i(k_1,\ldots, k_n) = 1.
    \end{align}
    For $i\neq j$, we have
    \begin{align}
        \Exp_{X_1,\ldots,X_n\sim\Binomial(m,p)}[L_iL_j] 
        &= 
        \sum_{k_1,\ldots,k_n=0}^m \left(\frac{p}{1-p}\right)^{2m-2k_i-2k_j} P(k_1,\ldots, k_n) \\
        &= \left(\frac{p}{1-p}\right)^{2m} \cdot \Exp_{X\sim\Binomial(m,p)}\left[\left(\frac{1-p}{p}\right)^{2X}\right]^2 \\
        &= \left(\frac{p}{1-p}\right)^{2m} \left[\left(\frac{1-p}{p}\right)^2 p + (1-p)\right]^{2m} \\
        &= 1.
    \end{align}
    This implies that the covariance $\mathrm{Cov}(L_i,L_j)=0$ for $i\neq j$. Moreover,
    \begin{align}
        \Exp_{X_1,\ldots,X_n\sim\Binomial(m,p)}[L_i^2]
        &= \sum_{k_1,\ldots,k_n=0}^m \left(\frac{(1-p)^{k_i}p^{m-k_i}}{p^{k_i}(1-p)^{m-k_i}}\right)^2 P(k_1,\ldots, k_n) \\
        &= \left(\frac{p}{1-p}\right)^{2m}\cdot \Exp_{X\sim\Binomial(m,p)}\left[ \left(\frac{1-p}{p}\right)^{4X} \right] \\
        &= \left(\frac{p}{1-p}\right)^{2m} \left[\left(\frac{1-p}{p}\right)^4 p + (1-p)\right]^m \\
        &= \left( \frac{3p^2-3p+1}{p(1-p)}\right)^m.
    \end{align}
    Combining these moments, we obtain
    \begin{align}
        \Exp_{X_1,\ldots,X_n\sim\Binomial(m,p)}\left[\left(\frac{1}{n}\sum_{i=1}^n L_i - 1\right)^2\right]
        &= \frac{1}{n^2} \left(\sum_{i=1}^n\mathrm{Var}(L_i) + \sum_{i\neq j} \mathrm{Cov}(L_i,L_j)\right) \\
        &= \frac{1}{n} \mathrm{Var}(L_1) \\
        &= \frac{1}{n}\left[\left(\frac{3p^2-3p+1}{p(1-p)}\right)^m - 1 \right].
    \end{align}
    Therefore, the total variation distance is bounded by
    \begin{align}
        \|P-Q\|_{TV}^2 
        &\leq \frac{1}{4} \Exp_{X_1,\ldots,X_n\sim\Binomial(m,p)}\left[\left(\frac{1}{n}\sum_{i=1}^n L_i - 1\right)^2\right] \\
        &\leq \frac{1}{4n}\left[\left(\frac{3p^2-3p+1}{p(1-p)}\right)^m - 1 \right].
    \end{align}
\end{proof}

It is worth noting that for $p\in(0,1)$, $0<p(1-p)\leq 1/4$, and thus
\begin{align}
    \frac{3p^2-3p+1}{p(1-p)} = \frac{1}{p(1-p)} - 3 \geq 1,
\end{align}
with equality if and only if $p=1/2$.
\thm{tv-bound} implies a lower bound on $m$ that scales logarithmically in $n$, stated in the following corollary.

\begin{corollary}\label{cor:n-lower-bound}
    For $p\in(0,1)\setminus\{1/2\}$ and $0<\delta<1/2$, every algorithm that solves \prob{hypothesis} with probability at least $1-\delta$ requires $m\geq \frac{\log(1+4n(1-2\delta)^2)}{\log\eta}$, where $\eta=\frac{1}{p(1-p)}-3>1$.
\end{corollary}
\begin{proof}
    By \thm{tv-bound}, the total variation distance between the two cases is at most $\left(\frac{1}{4n}(\eta^m - 1)\right)^{1/2}$.
    Any algorithm that solves \prob{hypothesis} with probability at least $1-\delta$ requires total variation distance at least $1-2\delta$, so we must have
    \begin{align}
        (1-2\delta)^2 \leq \frac{1}{4n}(\eta^m-1).
    \end{align}
    Rearranging the inequality gives $m\geq \log_\eta(1+4n(1-2\delta)^2)$.
\end{proof}

Combining \thm{delta-bound} and \cor{n-lower-bound}, we obtain the following sample complexity lower bound for the quantum state identity problem.

\begin{theorem}\label{thm:lower-bound}
    For $0<\epsilon<1/4$, constant $c>0$, and $0<\delta\le 1/2-c$, every quantum algorithm that solves \prob{hard} with parameter $\alpha=\frac{1}{2}\arccos\sqrt{1-\epsilon}$ (and hence every quantum algorithm that solves $\mathrm{QSI}_n^\epsilon$) with probability at least $1-\delta$ must use $\Omega\left(\frac{\log(n/\delta)}{\log(1/\epsilon)}\right)$ copies of each input state, where the constant hidden in the $\Omega(\cdot)$ depends only on $c$.
\end{theorem}
\begin{proof}
    By \thm{hard-qsi}, for $\alpha=\frac{1}{2}\arccos\sqrt{1-\epsilon}$, \prob{hard} is a special case of $\mathrm{QSI}_n^\epsilon$ with local dimension $d=2$; since qubit states can be embedded into any local dimension $d\ge 2$, any algorithm for $\mathrm{QSI}_n^\epsilon$ in local dimension $d$ also solves \prob{hard}.

    Let $\mathcal A$ be any quantum algorithm that uses $m$ copies of each input state and solves \prob{hard} with probability at least $1-\delta$.
    Then $\mathcal A$ must accept the average state
    \begin{align}
        \rho_0 = \Exp_{\theta\sim\mathrm{Uniform}(0,2\pi)^n}[\proj{\Psi_0(\theta)}^{\otimes m}]
    \end{align}
    with probability at least $1-\delta$ and accept the average state
    \begin{align}
        \rho_1 = \Exp_{i\sim\mathrm{Uniform}([n]),\theta\sim\mathrm{Uniform}(0,2\pi)^n}[\proj{\Psi_{1i}(\theta)}^{\otimes m}]
    \end{align}
    with probability at most $\delta$.
    After reordering the tensor factors so that the $m$ copies of each register are grouped together, we have $\ket{\Psi_0(\theta)}^{\otimes m}\cong\bigotimes_{i=1}^n\ket{\psi_{\theta_i}}^{\otimes m}$, and similarly for $\ket{\Psi_{1i}(\theta)}$. Since the phases $\theta_1,\ldots,\theta_n$ are independent, applying \lem{binomial} to each group of $m$ copies shows that $\rho_0$ and $\rho_1$ are both diagonal in the product basis $\{\ket{\wt_{k_1}}\otimes\ldots\otimes\ket{\wt_{k_n}}\}$, with diagonals given by the distributions of the yes and no cases of \prob{hypothesis}, respectively. Hence the trace distance between $\rho_0$ and $\rho_1$ equals the total variation distance in \prob{hypothesis}, and distinguishing $\rho_0$ from $\rho_1$ is exactly as hard as solving \prob{hypothesis}.

    Let $\eta$ be defined as in \cor{n-lower-bound}.
    Since $\cos 2\alpha=\sqrt{1-\epsilon}$, we have $p=\sin^2\alpha = \frac{1-\sqrt{1-\epsilon}}{2}$.
    This gives $p(1-p)=\epsilon/4$ and $\eta=4/\epsilon-3$, so that $\log\eta\le\log(4/\epsilon)\le 2\log(1/\epsilon)$ for $\epsilon<1/4$.
    By \thm{delta-bound} and \cor{n-lower-bound}, we obtain the two bounds
    \begin{align}
        m\geq \frac{\log(1/(4\delta(1-\delta)))}{\log(1/\epsilon)}
        \qquad\text{and}\qquad
        m\geq \frac{\log(1+4n(1-2\delta)^2)}{\log(4/\epsilon-3)}\geq\frac{\log(1+16c^2n)}{2\log(1/\epsilon)},
    \end{align}
    where we used $(1-2\delta)^2\ge 4c^2$. Both right-hand sides are positive, and together with $4\delta(1-\delta)\le 4\delta$ they give
    \begin{align}
        m=\Omega\left(\frac{\log(1/\delta)+\log n}{\log(1/\epsilon)}\right) = \Omega\left(\frac{\log(n/\delta)}{\log(1/\epsilon)}\right),
    \end{align}
    where the constant hidden in the $\Omega(\cdot)$ depends only on $c$.
\end{proof}

Under the promise of \cite{Kada_2008,buhrman2026permutationtestsquantumstate} that every pair of input states is either identical or orthogonal, a single copy of each input state suffices: the Permutation Test has perfect completeness and soundness error at most $1/n$.
In this paper, we analyze our lower bound for $\epsilon>0$, and show a nontrivial bound when $1/\epsilon$ grows no faster than polynomially in $n$ and $1/\delta$.

\section*{AI statement}

The results of this paper were developed through conversations with ChatGPT 5.6-Sol. 
The manuscript was written by the authors, drawing on the output of these conversations. 
The technical content was subsequently checked, and the manuscript proofread, with the help of Claude Fable 5.1 and ChatGPT 5.6-Sol, and the paper was revised to address the issues identified by these tools.
The authors take sole responsibility for the correctness of the results and for the content of the paper. 

\section*{Acknowledgments}

We acknowledge the support by National Science and Technology Council in Taiwan, Grant~No.~114-2222-E-002-002-MY3, and by National Taiwan University, Grant~No.~NTU-CC-114L895503.

\bibliographystyle{plainurl}
\bibliography{refs}

\begin{thebibliography}{10}

\bibitem{baconQuantumSchurClebschGordan2007}
Dave Bacon, Isaac~L. Chuang, and Aram~W. Harrow.
\newblock The quantum {{Schur}} and {{Clebsch-Gordan}} transforms: {{I}}.
  efficient qudit circuits.
\newblock In {\em Proceedings of the Eighteenth Annual {{ACM-SIAM}} Symposium
  on {{Discrete}} Algorithms}, {{SODA}} '07, pages 1235--1244, USA, January
  2007. {Society for Industrial and Applied Mathematics}.
\newblock arXiv:quant-ph/0601001.
\newblock \href {https://arxiv.org/abs/quant-ph/0601001}
  {\path{arXiv:quant-ph/0601001}}.

\bibitem{10.1098/rspa.2022.0733}
Zachary~P. Bradshaw, Margarite~L. LaBorde, and Mark~M. Wilde.
\newblock Cycle index polynomials and generalized quantum separability tests.
\newblock {\em Proceedings of the Royal Society A: Mathematical, Physical and
  Engineering Sciences}, 479(2274):20220733, June 2023.
\newblock arXiv:2208.14596.
\newblock \href {https://arxiv.org/abs/2208.14596} {\path{arXiv:2208.14596}},
  \href {https://doi.org/10.1098/rspa.2022.0733}
  {\path{doi:10.1098/rspa.2022.0733}}.

\bibitem{bruinsma2026quantumfouriertransformsymmetric}
Carli Bruinsma, Dmitry Grinko, and Maris Ozols.
\newblock Quantum {Fourier} transform for the symmetric group.
\newblock arXiv:2608.28569, 2026.
\newblock URL: \url{https://arxiv.org/abs/2608.28569}, \href
  {https://arxiv.org/abs/2608.28569} {\path{arXiv:2608.28569}}, \href
  {https://doi.org/10.48550/arXiv.2608.28569}
  {\path{doi:10.48550/arXiv.2608.28569}}.

\bibitem{buhrman2026permutationtestsquantumstate}
Harry Buhrman, Dmitry Grinko, Philip {Verduyn Lunel}, and Jordi Weggemans.
\newblock Permutation tests for quantum state identity.
\newblock arXiv:2405.09626, 2026.
\newblock URL: \url{https://arxiv.org/abs/2405.09626}, \href
  {https://arxiv.org/abs/2405.09626} {\path{arXiv:2405.09626}}, \href
  {https://doi.org/10.48550/arXiv.2405.09626}
  {\path{doi:10.48550/arXiv.2405.09626}}.

\bibitem{burchardtHighdimensionalQuantumSchur2025}
Adam Burchardt, Jiani Fei, Dmitry Grinko, Martin Larocca, Maris Ozols, Sydney
  Timmerman, and Vladyslav Visnevskyi.
\newblock High-dimensional quantum {{Schur}} transforms.
\newblock arXiv:2509.22640, September 2025.
\newblock \href {https://arxiv.org/abs/2509.22640} {\path{arXiv:2509.22640}},
  \href {https://doi.org/10.48550/arXiv.2509.22640}
  {\path{doi:10.48550/arXiv.2509.22640}}.

\bibitem{cincio2018learning}
Lukasz Cincio, Yi{\u{g}}it Suba{\c{s}}i, Andrew~T. Sornborger, and Patrick~J.
  Coles.
\newblock Learning the quantum algorithm for state overlap.
\newblock {\em New Journal of Physics}, 20(11):113022, 2018.
\newblock arXiv:1803.04114.
\newblock \href {https://arxiv.org/abs/1803.04114} {\path{arXiv:1803.04114}},
  \href {https://doi.org/10.1088/1367-2630/aae94a}
  {\path{doi:10.1088/1367-2630/aae94a}}.

\bibitem{garciaescartin2013swap}
Juan~Carlos Garcia-Escartin and Pedro Chamorro-Posada.
\newblock {SWAP} test and {Hong-Ou-Mandel} effect are equivalent.
\newblock {\em Physical Review A}, 87(5):052330, 2013.
\newblock arXiv:1303.6814.
\newblock \href {https://arxiv.org/abs/1303.6814} {\path{arXiv:1303.6814}},
  \href {https://doi.org/10.1103/PhysRevA.87.052330}
  {\path{doi:10.1103/PhysRevA.87.052330}}.

\bibitem{harrow2005applicationscoherentclassicalcommunication}
Aram~W. Harrow.
\newblock Applications of coherent classical communication and the {Schur}
  transform to quantum information theory.
\newblock arXiv:quant-ph/0512255, 2005.
\newblock PhD thesis, Massachusetts Institute of Technology.
\newblock URL: \url{https://arxiv.org/abs/quant-ph/0512255}, \href
  {https://arxiv.org/abs/quant-ph/0512255} {\path{arXiv:quant-ph/0512255}},
  \href {https://doi.org/10.48550/arXiv.quant-ph/0512255}
  {\path{doi:10.48550/arXiv.quant-ph/0512255}}.

\bibitem{Kada_2008}
Masaru Kada, Harumichi Nishimura, and Tomoyuki Yamakami.
\newblock The efficiency of quantum identity testing of multiple states.
\newblock {\em Journal of Physics A: Mathematical and Theoretical},
  41(39):395309, September 2008.
\newblock arXiv:0809.2037.
\newblock URL: \url{http://dx.doi.org/10.1088/1751-8113/41/39/395309}, \href
  {https://arxiv.org/abs/0809.2037} {\path{arXiv:0809.2037}}, \href
  {https://doi.org/10.1088/1751-8113/41/39/395309}
  {\path{doi:10.1088/1751-8113/41/39/395309}}.

\bibitem{pelecanos2025debiasedkeylsalgorithmnew}
Angelos Pelecanos, Jack Spilecki, and John Wright.
\newblock The debiased {Keyl}'s algorithm: a new unbiased estimator for full
  state tomography.
\newblock arXiv:2510.07788, 2025.
\newblock URL: \url{https://arxiv.org/abs/2510.07788}, \href
  {https://arxiv.org/abs/2510.07788} {\path{arXiv:2510.07788}}, \href
  {https://doi.org/10.48550/arXiv.2510.07788}
  {\path{doi:10.48550/arXiv.2510.07788}}.

\bibitem{wright2016learn}
John Wright.
\newblock {\em How to learn a quantum state}.
\newblock PhD thesis, Carnegie Mellon University, 2016.

\end{thebibliography}

\end{document}